\documentclass[onecolumn]{IEEEtran}

\usepackage[dvips]{graphicx}
\usepackage{latexsym}
\usepackage{epsfig}
\usepackage{color}
\usepackage{amsmath} 
\usepackage{amssymb}
\usepackage{upgreek} 
\usepackage{amsxtra}
\usepackage{enumitem}
\usepackage{float}
\usepackage[T1]{fontenc}
\usepackage{setspace}
\usepackage{balance}
\usepackage{diagbox}
\usepackage{epstopdf}
\usepackage{amsthm}
\usepackage{dsfont} 
\usepackage{multirow}
\usepackage{hyperref}
\usepackage{tabularx,colortbl}
\usepackage[table]{xcolor}
\usepackage{lscape}
\usepackage{url}
\usepackage{cite}
\usepackage{algorithm}
\usepackage{algpseudocode}
\usepackage[ subrefformat=parens,labelformat=parens, caption=false, font=footnotesize]{subfig}
\usepackage{mathtools}
\usepackage{units}
\usepackage{cases}

\newcommand{\Tpow}{{\sf T}}
\newcommand{\Invpow}{{\sf -1}}
\newcommand{\Strpow}{{\sf *}}
\newcommand{\Hpow}{{\sf H}}

\newtheorem{lemma}{\textbf{Lemma}}

\newcommand{\abs}[1]{\left|{#1}\right|}

\newcommand{\nth}[1]{{#1}{\text{th}}}

\begin{document}

\title{Symbol Error Analysis of Linear Receivers in Terahertz Channels under Channel-Noise Dependence
}

\author{
        Almutasem Bellah Enad, \IEEEmembership{Graduate Student Member, IEEE,}
        Jihad Fahs, \IEEEmembership{Senior Member, IEEE,}
        Hadi Sarieddeen, \IEEEmembership{Senior Member, IEEE,}
        Hakim Jemaa, \IEEEmembership{Graduate Student Member, IEEE}
        and Tareq Y. Al-Naffouri, \IEEEmembership{Fellow, IEEE}%
\thanks{ 
Almutasem~Bellah~Enad, Jihad Fahs, and Hadi Sarieddeen, are with the American University of Beirut (AUB), Beirut, Lebanon. (aae118@mail.aub.edu, \{jihad.fahs, hadi.sarieddeen\}@aub.edu.lb).
 Hakim Jemaa and Tareq Y. Al-Naffouri are with King Abdullah University of Science and Technology (KAUST), Kingdom of Saudi Arabia, \{hakim.jemaa, tareq.alnaffouri\}@kaust.edu.sa.
 This work is supported by the AUB's University Research Board and Vertically Integrated Projects Program, and KAUST's Office of Sponsored Research under Award No.~ORFS-CRG12-2024-6478.}
}

\maketitle

\begin{abstract} 
This paper develops a comprehensive framework for the performance analysis of linear detectors, namely zero-forcing (ZF) and minimum mean-square error (MMSE), under diverse terahertz (THz) channel conditions. Three fading models are considered: Rayleigh fading, the $\alpha$--$\mu$ distribution for indoor THz environments, and the mixture-gamma (MG) distribution for outdoor THz scenarios. Semi-analytical, approximate, and asymptotic expressions for the symbol error rate (SER) are derived, explicitly incorporating the correlation between the channel and the additive noise arising from hardware impairments. This correlation is characterized using both statistical approaches and copula-based methods to effectively capture complex dependency structures. The theoretical findings are validated through simulations, demonstrating strong agreement with the derived expressions and confirming the accuracy and robustness of the proposed framework. The results demonstrate the significant impact of channel--noise dependence on THz-band receiver performance and verify the expected performance degradation of biased MMSE receivers in point-to-point links employing higher-order quadrature amplitude modulation. Specifically, at a target SER of $10^{-3}$, a 70\% correlation results in approximately a 6.5~dB degradation in the effective signal-to-noise ratio, with mismatched MMSE detection incurring an additional 1~dB loss compared to ZF. Nonetheless, MMSE offers enhanced numerical stability under severe channel fading conditions, where channel inversion causes noise amplification.

\end{abstract}


\section{Introduction}
\IEEEPARstart{T}{he} terahertz (THz) band, spanning 100 gigahertz (GHz) to 10 THz, is a key enabler for future wireless communications. THz communications promise terabit-per-second (Tbps) data rates and ultra-low latency, enabling high-speed backhaul, extended reality, and wireless sensing applications~\cite{Jornet2024Evolution,Sarieddeen9514889}. However, challenges in THz signal propagation and transceiver design hinder the full potential of this technology~\cite{Sarieddeen9514889}. THz channels are highly sparse in time and space because of severe attenuation, molecular absorption, strong directionality, and blockage sensitivity, though non-line-of-sight (NLoS) communication remains feasible, particularly indoors with high-gain antennas~\cite{9591285}. The need for low-complexity, Tbps-achieving baseband operations further constrains THz-band signal processing design~\cite{Sarieddeen2024Bridging}. Robust and realistic THz performance analysis frameworks are therefore critical, yet many existing models overlook practical effects such as absorption, misalignment, and the noise-coloring induced by linear baseband filtering~\cite{Sarieddeen2024Bridging,8610080,9714471,electronics12061336,10663782}.

Accurate channel and noise modeling is essential for system optimization and forms the foundation of the proposed analysis, particularly through the use of the $\alpha$--$\mu$ and mixture-gamma (MG) distributions for indoor and outdoor THz fading, respectively~\cite{8610080,papasotiriou2021experimentally}. The $\alpha$-$\mu$ distribution is widely recognized for its ability to capture indoor THz small-scale fading due to its flexible parameterization~\cite{papasotiriou2021experimentally}, and has been extensively applied in the development of analytical frameworks for indoor THz channels~\cite{papasotiriou2021experimentally,11072420,11362947}. Although recent studies have proposed more comprehensive models, further validation under realistic operating conditions remains necessary. For example, \cite{8610080} incorporates misalignment fading through the zero-boresight model, and \cite{9714471} extends this framework by integrating random fog effects modeled with a gamma distribution. In outdoor scenarios, \cite{bhardwaj2022performance} develops an analytical framework with closed-form expressions that leverage the generalized Gaussian mixture (GM) distribution for small-scale fading, \cite{electronics12061336} investigates atmospheric channel effects using a gamma-gamma distribution tailored for outdoor environments, and \cite{10663782} introduces an analytical framework for single-input single-output (SISO) outdoor THz channels, incorporating MG distribution and misalignment fading.

As conventional assumptions of channel–noise independence can break down in emerging technologies such as THz-band systems and reconfigurable intelligent surfaces (RISs)~\cite{Sarieddeen9514889,10663782,8610080}, accounting for channel–noise correlation becomes essential for accurate performance analysis. In practical THz systems, such dependence may arise from channel-induced molecular absorption noise (caused by retransmissions following molecular absorption at high frequencies)~\cite{Sarieddeen9514889} or hardware imperfections~\cite{8610080}, for example. Since traditional independent-noise models cannot capture these effects~\cite{10663782,Magableh2009,9714471}, incorporating channel–noise correlation provides a more general and physically consistent framework. Recent studies by the authors~\cite{11072420,11078009} further highlight this need by analyzing the system performance under correlated channels and noise, underscoring the importance of addressing this gap in next-generation network design. Copula theory has been widely applied in wireless communications to model complex dependencies. For example, copulas have been used to model Nakagami-$m$ fading with tail dependence~\cite{4487493} and to capture non-standard dependencies between frequency bands in fading channels~\cite{7032317}. They have also been employed to estimate channel parameters and determine the capacity of correlated multiple-input multiple-output (MIMO) systems~\cite{Gholizadeh2015OnTC} and to study the impact of correlated fading on multiple-access channels~\cite{9762969}. More recently, copulas have been used for THz channel modeling~\cite{11072420} and further applied to analyze RIS-aided communications with phase noise, where the Farlie-Gumbel-Morgenstern (FGM) copula has mainly been used to model correlation over Rayleigh fading links~\cite{9690184,10453225,enad2025theoretical}. Despite extensive work on THz performance analysis, the impact of channel--noise dependence on linear receiver performance remains largely unexplored, particularly across diverse THz fading models and practical detectors.

At the link level, linear detectors such as zero-forcing (ZF) and minimum mean-square error (MMSE) are near-optimal in SISO links and are also preferred in MIMO scenarios for their computational efficiency, which is crucial when Tbps data rates are sought~\cite{Sarieddeen9514889}. Most symbol error rate (SER) analyses of linear detection schemes assume Gaussian noise after linear filtering, such as post-multiplication by the channel pseudo-inverse in ZF detection. However, deriving the probability density function (PDF) of the post-detection effective noise is essential for accurate SER analysis in fast-fading channels, especially when the noise is channel-dependent. This characterization is also important when detection is followed by noise-aware channel decoding, where post-detection effective noise can provide pseudo-soft information for decoding~\cite{Sarieddeen2022GRAND}.

In point-to-point SISO links, the ZF detector is considered optimal in high-signal-to-noise ratio (SNR) regimes due to its ability to completely remove the channel effect; however, its performance degrades under poor channel conditions~\cite{Sarieddeen9514889}. Specifically, when channel magnitude approaches zero, the ZF detector amplifies the noise power unboundedly, potentially causing severe error propagation and even quantizer overflow in practical receivers~\cite{Padgett2009}. In contrast, the MMSE detector regularizes the channel inversion through a scaling factor, limiting noise enhancement and stabilizing the received signal, which makes MMSE more robust in deep fading scenarios~\cite{4544952,proakis2008digital}. However, such regularization introduces a bias that alters the recovered symbols, explaining the inferior performance of MMSE compared to ZF. When both detectors operate under the same decision regions, the performance of the MMSE detector is no longer equivalent to that of ZF  at high SNR, since the inherent bias shifts the decision boundaries relative to ZF. This bias motivates a dedicated SER analysis of MMSE detection under ZF-based decision regions (a common practice) to accurately capture its performance characteristics.

\subsection{Contributions}
This paper presents a comprehensive framework for analyzing the SER of linear data detection in SISO THz-band systems under various channel conditions. The system assumes fast-fading channels with additive Gaussian noise, where the channel and noise terms may exhibit statistical dependence, addressing practical cases where such dependence arises. The framework is validated through the SER analysis of ZF and MMSE detectors employing quadrature amplitude modulation (QAM), demonstrating high accuracy across different THz channel models. Two correlated modeling approaches are adopted: (i) direct statistical correlation and (ii) copula-based methods~\cite{nelsen2006introduction}. In the independent case, the framework reduces to conventional models~\cite{4544952}, while providing a generalized solution for dependent scenarios. Extensive numerical results demonstrate that the proposed framework accurately characterizes the SER across a wide range of THz operating conditions, including cases with hardware impairments~\cite{8610080}. Moreover, the framework's flexible structure allows straightforward extensions to other channel and noise models, making it a valuable tool for the design and optimization of THz communication systems.
The main technical contributions are:
\begin{itemize}
\item Building upon our prior work on ZF detection under Rayleigh fading~\cite{11078009} and indoor THz channels~\cite{11072420}, we extend the analysis to outdoor THz fading environments. We derive the PDF of the ratio between a Gaussian random variable and an MG-distributed random variable, and leverage this result to derive SER expressions under correlated conditions. Channel--noise dependence is modeled using a copula-based framework. Furthermore, we consider a realistic correlated THz system model and derive the corresponding effective-noise PDF.

\item We investigate the performance of MMSE detection over Rayleigh, indoor THz, and outdoor THz fading channels by deriving the PDF of the post-filtered received signal under general channel--noise correlation models. MMSE introduces an inherent estimation bias when conventional ZF-based decision regions are employed, as is commonly assumed in practice~\cite{tse2005fundamentals}. This observation motivates a dedicated SER analysis of SISO MMSE detection.

\item We derive SER expressions in multiple forms, including novel semi-analytical formulations, closed-form solutions, asymptotic characterizations, and accurate approximations. We demonstrate analytical tractability and numerical accuracy across a broad range of THz channels.

\end{itemize}

\subsection{Organization and notation}
The remainder of the paper is organized as follows. Section~\ref{sec:system_model} introduces the system and channel models. The SER analysis is presented in Section~\ref{sec:problem_formulation}. We provide the error probability expressions for the Rayleigh channel in Section~\ref{sec:SER_gaussian} and for the indoor and outdoor THz channels in Section~\ref{sec:SER_THz}. Simulation results are discussed in Section~\ref{sec:simultaion}, and concluding remarks are given in Section~\ref{sec:conclusion}. 
We employ the following mathematical notation: the $\abs{\cdot}$ operator represents the absolute value, $\mathbb{E}[\cdot]$ denotes the expectation operator, and $\mathbb{P}(\cdot)$ signifies the probability operator. The superscripts ${(\cdot)}^{\Tpow}$, ${(\cdot)}^\Strpow$, ${(\cdot)}^\Hpow$, and ${(\cdot)}^\Invpow$ represent the transpose, conjugate, Hermitian, and inverse operators, respectively. The circularly symmetric complex Gaussian distribution with independent and identically distributed real and imaginary components, each with variance $\sigma^2$, is denoted by $\mathcal{CN}(0,2 \sigma^2)$. The gamma function is denoted by $\Gamma(.)$, and the incomplete gamma function is defined as $\gamma(\zeta, u)=\int_0^u t^{\zeta-1} e^{-t} d t$. The $Q$-function, given by $Q(x)=\frac{1}{\sqrt{2 \pi}} \int_x^{\infty} e^{ -\frac{u^2}{2}} d u$, and the error function, given by $\operatorname{erf}(x)=\frac{2}{\sqrt{\pi}} \int_0^{x} e^ {-u^2} d u$, are used extensively in performance analysis~\cite{prudnikov1986integrals}. We further use the Meijer G-function, represented as $G_{p, q}^{m, n}\left[z \left\lvert \begin{smallmatrix}a_1, \ldots, a_p \\ b_1, \ldots, b_q\end{smallmatrix} \right.\right]$~\cite[eq. (9.301)]{prudnikov1986integrals,kilbas2004h}.

\section{System Model}\label{sec:system_model}
We conduct a comparative analysis over multiple system models: a reference Rayleigh channel model, a generic indoor/outdoor THz channel model that captures fast-fading effects, a specialized THz model that accounts for unique propagation losses and distortion noise caused by hardware impairments, and a generic model of noise–channel correlation using copulas. By evaluating the impact of THz-specific factors on system performance, this study provides valuable insights into the robustness of wireless communications under different THz-band fading conditions.

\subsection{Reference Rayleigh Channel}
We consider a fast-fading, narrow-band SISO communication system with a frequency-domain input-output relation. To account for the correlation between the channel and the noise, we adopt the system model:
\begin{align}
r &= hs + n \notag\\
&= hs + \sqrt{2}\,\lambda^* \sigma_n h + \sqrt{2(1 - |\lambda|^2)}\sigma_n \tilde{n}\notag \\ &= h(s + \sqrt{2}\lambda^* \sigma_n) + \sqrt{2(1 - |\lambda|^2)}\sigma_n \tilde{n},\label{generic_sys}
\end{align}
where $r$ denotes the received symbol, $s$ is the transmitted QAM symbol, and $h$ and $n$ are circularly symmetric Gaussian with zero mean and variances $\mathbb{E}[hh^*] = 1$ and $\mathbb{E}[nn^*] = 2\sigma_n^2$, representing the channel and the additive noise, respectively. \(\lambda = \lambda_x + j\lambda_y\) represents the complex correlation coefficient between \(h\) and \(n\). The noise \(\tilde{n} \sim \mathcal{CN}(0,1)\) ($n\!=\!\sqrt{2}\sigma_n\, \tilde{n}$) is independent of \(h\). The PDF and CDF of $|h|$ and $|\tilde{n}|$ are 
\begin{equation}
    f_{X}(x) \!=\! \frac{x}{\sigma_X^2} \exp\left( -\frac{x^2}{2\sigma_X^2} \right)\ \text{and} \ 
    F_{X}(x) \!=\!1 - \exp\left(\!-\dfrac{x^2}{2\sigma_X^2}\right)\!, \label{PDF_n_h_ray}
\end{equation}
respectively, where $2\sigma_X^2 \!=\!1$. The distributions of $|h|^2$ and $|\tilde{n}|^2$ follow an exponential distribution (Chi-squared with two degrees of freedom),
\begin{equation}
    f_{X^2}(x) = \frac{1}{\sigma_X^2} e^{-\tfrac{x}{\sigma_X^2}}, \quad x \geq 0. \label{PDF_|h|2_GS}
\end{equation}

\subsection{Generic THz System Model}

We consider a system of baseband input-output relation,
\begin{equation}
    r=\sqrt{p_tG_tG_r} h_p h_f s + n = \nu h_f s + n = h s + n, \label{generic_Thz_sys}
\end{equation}
where $p_t$ is the average transmit power, $G_t$ and $G_r$ are the transmit and receive antenna gains, respectively, and where we have $h\!=\!\nu h_f$ and $\nu \triangleq \sqrt{p_tG_tG_r} h_p$. The additive noise, $n \sim \mathcal{CN}(0,2 \sigma_n^2)$, has a Rayleigh-distributed amplitude with a PDF and CDF given in~\eqref{PDF_n_h_ray}.

For the channel, $h_p$ represents the THz-band free-space path loss, consisting of both spreading and molecular absorption losses, expressed as~\cite{9591285}
\begin{equation}
    h_p = \left(c/(4 \pi f d)\right)^{\frac{\varrho}{2}} e^{-\frac{1}{2} K_{\text{abs}} d}, \notag
\end{equation}
where $c$ is the speed of light, $f$ is the operating frequency, $d$ is the communication distance, and $K_{\text{abs}}$ is the molecular absorption coefficient (more details in \cite{9591285}). In measurement-based sub-THz/THz works \cite{9591285}, \cite{10663782}, the path loss exponent, $\varrho$, is best-fit to 2. Furthermore, in the case of the indoor THz model, $h_f$ represents complex small-scale fading of $\alpha$-$\mu$-distributed magnitude \cite{Magableh2009}. Thus, the PDF and CDF of $|h|$ can be expressed as, 
\begin{equation}
    f_{|h|}(y) = \frac{1}{|\nu|}f_{|h_f|}\left(\frac{y}{\nu}\right) = \frac{\alpha \mu^\mu y^{\alpha \mu - 1}}{(\hat{Z}\nu)^{\alpha \mu} \Gamma(\mu)} \hspace{-0.5pt} \exp\left(\!-\mu y^\alpha/(\nu\hat{Z})^{\alpha}\!\right), \label{PDF_|h|_AM}
\end{equation}
\begin{equation}
    F_{|h|}(y)=\frac{\gamma{(\mu,\mu y^\alpha / (\hat{Z}\nu)^\alpha)}}{\Gamma(\mu)}, \ \text{for} \ \, y \geq 0, \label{CDF_|h|_AM}
\end{equation}
where $\alpha\!>\! 0$ is a fading parameter, $\mu$ is the normalized variance of the fading channel, and $\hat{Z} = \sqrt[\alpha]{\mathbb{E}(|h_f|^\alpha)}$ is the $\alpha$ root mean value of the fading channel. Using equation~\eqref{PDF_|h|_AM}, the PDF of \( |h|^2 \) is given by 
\begin{equation}
    f_{|h|^2}(y) = \frac{\alpha \mu^\mu}{2 (\hat{Z}\nu)^{\alpha \mu} \Gamma(\mu)} y^{\frac{\alpha \mu }{2}-1} \exp\left(\!-\mu y^{\alpha/2} / (\nu \hat{Z})^\alpha \!\right).\label{PDF_|h|2_AM}
\end{equation}
For the outdoor THz model, the MG distribution is used to describe the magnitude of small-scale fading ~\cite{papasotiriou2023outdoor}, where the PDF and CDF are expressed as~\cite{papasotiriou2023outdoor}
\begin{equation}
    f_{|h|}(y)=\sum_{i=1}^K \! \frac{w_i}{\nu^{\beta_i}} \frac{\zeta_i^{\beta_i} y^{\beta_i-1} e^{-\frac{\zeta_i}{\nu} y}}{\Gamma\left(\beta_i\right)}= \sum_{i=1}^{K} \frac{\alpha_i}{\nu^{\beta_i}} y^{\beta_i - 1} e^{-\frac{\zeta_i}{\nu} y}, \label{PDF_|h|_MG}
\end{equation}
\begin{equation}
    F_{|h|}(y) = \sum_{i=1}^{K} \alpha_i \zeta_i^{-\beta_i} \gamma\left(\beta_i,\frac{\zeta_i}{\nu} y\right),\label{CDF_|h|_MG}
\end{equation}
where $K$ is the number of gamma components, and $\zeta _i$, $\beta _i$ and $w_i$ denote the scale, shape, and weight of the $\nth{i}$ component; $\sum_{i=1}^K w_i \!=\! 1$. We define $\alpha_i\!=\!w_i\zeta_i^{\beta_i}/\Gamma\left(\beta_i\right)$ for ease of notation. 
Similarly, the PDF of \( |h|^2 \) can be expressed as
\begin{equation}
f_{|h|^2}(y) = \sum_{i=1}^{K} \frac{\alpha_i}{2 \nu^{\beta_i}} y^{\frac{\beta_i}{2}-1} e^{-\frac{\zeta_i}{\nu} \sqrt{y}}.\label{PDF_|h|2_MG}
\end{equation}

\subsection{Copula Correlation Model}\label{copula}
 When the joint distribution of the channel and noise cannot be accurately described by classical measures (e.g., the Pearson correlation between complex Gaussian noise and $\alpha$-$\mu$ fading), we adopt a copula-based statistical model~\cite{nelsen2006introduction}, as applied in THz systems~\cite{10453225,11072420}. The joint PDF of $|n|$ and $|h|$ is expressed as  
\begin{equation}
 f_{|n|,|h|}(x,y) = C(F_{|n|}(x), F_{|h|}(y)) f_{|n|}(x) f_{|h|}(y), \label{copula_PDF}
\end{equation}
where $f_{|n|}(\cdot)$, $f_{|h|}(\cdot)$ and $F_{|n|}(\cdot)$, $F_{|h|}(\cdot)$ represent the marginal PDFs and CDFs, and $C(\cdot,\cdot)$ is the copula density.  
We consider two copula families: the Farlie-Gumbel-Morgenstern (FGM) which models weak dependence,  
\begin{equation}
    C(u,v)=1+\lambda(2u-1)(2v-1), \quad \lambda\in[-1,1], \label{FGM}
\end{equation}
and the Frank copula model for stronger dependencies,  
\begin{equation}
    C(u,v)=\frac{-\lambda(e^{-\lambda}-1)e^{-\lambda(u+v)}}{\left[(e^{-\lambda u}-1)(e^{-\lambda v}-1)+(e^{-\lambda}-1)\right]^2}, \quad \lambda\neq0,\label{FRANK}
\end{equation}
where $\lambda$ controls the dependence strength, with $\lambda>0$ indicating positive dependence. This framework is used to characterize the channel-noise correlation in system model~\eqref{generic_Thz_sys}.

\subsection{Specific THz System Model}
We consider a model that incorporates THz distortion noise from hardware impairments, capturing the correlation between channel fading and receiver noise in \cite{8610080,11072420},
\begin{align}
    r &= \sqrt{p_tG_tG_r} h_p h_f (s + n_t) + n_r + n,\notag\\&= h(s + n_t) + n_r + n,\notag\\& 
    =hs+n_{tr}+n.
    \label{Sp_Thz_sys}
\end{align}
Given $h$, the distortion noise components due to transmitter and receiver impairments are modeled as \( n_t \sim \mathcal{CN}(0, \kappa_t^2) \) and \( n_r \sim \mathcal{CN}(0, \kappa_r^2 |h|^2) \), respectively. The additive noise term \( n \sim \mathcal{CN}(0, 2\sigma_n^2) \) is independent of \( h \). Additionally, the combined distortion noise is given by \( n_{tr} \sim \mathcal{CN}(0, 2\sigma_{tr}^2) \), where \(2\sigma_{tr}^2 = \kappa_r^2 |h|^2 + \kappa_t^2 |h|^2 = 2\bar{\sigma}_{tr} |h|^2,\) with \( 2\bar{\sigma}_{tr} = \kappa_r^2 + \kappa_t^2 \). 
The parameters, $\kappa_t$ and $\kappa_r$, are non-negative design factors representing hardware imperfections. Here, the PDF of $\abs{h}$ is defined in~\eqref{PDF_|h|_AM} and \eqref{PDF_|h|_MG} for indoor and outdoor channels, respectively.

\section{Error Analysis for Linear Detectors}\label{sec:problem_formulation}


We highlight the analytical foundation of the SER evaluation under THz fading channels in four parts: (i) SER analysis for ZF, (ii) SER analysis for MMSE, (iii) asymptotic SER characterization, and (iv) post-filtering detection. 

\subsection{SER for ZF Detector}
At the receiver, a ZF detector inverts the channel effect as
\begin{equation}
    \tilde{r} = h^{-1}r = s+\frac{n}{h} = s + z,\label{effective_zf}
\end{equation}
where effective noise, $z\!=\!n/h$, is the ratio of two dependent variables due to channel-induced molecular absorption noise\cite{Sarieddeen9514889}. In~\cite{11072420,11078009}, we developed a generic framework for calculating the SER that does not assume channel-independent noise, by leveraging the effective-noise ratio distribution. The average SER is expressed as
\begin{align}
    &\text{SER}= 1 - \frac{1}{M}\sum_{k=1}^{M} P_{c|s_k}, \notag 
    \\&P_{c|s_k} =  \iint_{Dk} f_{z_x,z_y}(\tilde{r}_x-s_x, \tilde{r}_y - s_y)\, d\tilde{r}_x\, d\tilde{r}_y,\label{SER-ZF}
\end{align}
where $P_{c|s_k}$ is the probability of correct detection of transmitted symbol \(s_k = (s_x,s_y)\), $\tilde{r}_x$ and $\tilde{r}_y$ denote the real and imaginary parts of $\tilde{r}$, $D_k$ represents the decision region of symbol $s_k$, and $f_{z_x,z_y}(\cdot,\cdot)$ is the PDF of the effective noise. 

\subsection{SER for MMSE Detector}\label{MMSE_SER_fram}
For a linear observation model, $r = h s + n$, the MMSE detector is given by $\mathbb{E}[s \!\mid\! r] = \frac{\text{Cov}(s,r)}{\text{Var}(r)} \, r$~\cite{proakis2008digital,4544952}. 
Assuming that $s$ and $n$ are zero-mean and independent, we get
\begin{eqnarray*}
\text{Cov}(s,r) &=& \mathbb{E}[s r^*] - \mathbb{E}[s]\mathbb{E}[r^*] =  h^* \mathbb{E}[|s|^2] = h^* \sigma_s^2\\
\text{Var}(r) &=& \mathbb{E}[|r|^2] = |h|^2 \sigma_s^2 + 2 \sigma_n^2.
\end{eqnarray*}
Letting \(\tilde{\sigma}\triangleq\frac{\sqrt{2}\sigma_n}{\sigma_s}\), then the received signal post filtering is
\begin{equation}
    \tilde{r}  
    = \frac{h^* }{|h|^2 + \tilde{\sigma}^2} r = \frac{|h|^2}{|h|^2 + \tilde{\sigma}^2} s + \frac{h^*}{|h|^2 + \tilde{\sigma}^2} n,\label{MMSE_filter}
\end{equation}

For analysis, we first derive the PDF of \(\tilde{r}\) and then evaluate the SER within the decision region of the symbol \(s = (s_x,s_y)\), which is readily a minimum distance region. 
Given $|h|$, \(\tilde{r} = (\tilde{r}_x,\tilde{r}_y)\) will be complex Gaussian,
\begin{equation*}
   \tilde{r}_{|h} \sim \mathcal{CN}\left(\frac{|h|^2}{|h|^2 + \tilde{\sigma}^2} s, \frac{|h|^2 }{(|h|^2 + \tilde{\sigma}^2)^2}2\sigma_n^2\right),
\end{equation*}
and the PDF of \(\tilde{r}\) is given by
\begin{equation}
    f_{\tilde{r}}(\tilde{r})=\int_0^\infty p_{\tilde{r}_{|h}}(t,\tilde{r})  f_{|h|^2}(t) \,dt.
    \label{eq:tilder}
\end{equation}
Equation~\eqref{eq:tilder} boils down to:
\begin{align}
 f_{\tilde{r}}&(\tilde{r}_x,\tilde{r}_y)= \frac{1}{\pi } \int_0^\infty \frac{(t + \tilde{\sigma}^2)^2}{2t\sigma_n^2} \exp\left(- \frac{(t + \tilde{\sigma}^2)^2}{2t \sigma_n^2} \left[\left(\tilde{r}_x - \frac{t}{t + \tilde{\sigma}^2} s_x\right)^2 + \left(\tilde{r}_y - \frac{t}{t + \tilde{\sigma}^2} s_y\right)^2\right]\right)  f_{|h|^2}(t) dt.\label{PDF_MMSE}
\end{align}
The probability of correct detection, $P_{c|s_k}$, $1 \leq k \leq M$, is computed as
\begin{equation}
    P_{c|s_k}=\iint_{D_k} f_{\tilde{r}}(\tilde{r}_x,\tilde{r}_y)\, d\tilde{r}_x\, d\tilde{r}_y, \label{eq:error_PDF}
\end{equation}
Substituting~\eqref{PDF_MMSE} in~\eqref{eq:error_PDF},
\begin{align*}
     P_{c|s_k}&=\frac{1}{\pi}\int_0^\infty\!\!\!\! A     \underbrace{\left(\int_{b^{(1)}_k}^{b^{(2)}}\exp\left(-A\left(\tilde{r}_y\!-\!Bs_y\right)^2\right)d\tilde{r}_y\right)}_{I_1 \, \text{integral}}   \underbrace{\left(\int_{b_k^{(3)}}^{b_k^{(4)}}\exp\left(-A\left(\tilde{r}_x\!-\!Bs_x\right)^2\right)d\tilde{r}_x\right)}_{I_2 \, \text{integral}} f_{|h|^2}(t)dt.
\end{align*}
where \( A = A(t) \triangleq \frac{(t+\tilde{\sigma}^2)^2}{2t\sigma_n^2} \), \( B = B(t) \triangleq \frac{t}{t+\tilde{\sigma}^2}\), and for an M-QAM constellation, the symbol $s_k$ has the corresponding decision region \( D_k = [b_k^{(1)},b_k^{(2)}] \times [b_k^{(3)},b_k^{(4)}] \). Since the integrals $I_1$ and $I_2$ are similar in form, we solve only \( I_1 \). Letting $u = \sqrt{A}(\tilde{r}_y - B s_y)$ transforms the integral into
\begin{eqnarray*}
I_1 \!\!\!\!\!&=&\!\!\!\!\! \frac{1}{\sqrt{A}} \int_{\sqrt{A}(b^{(1)}_k - B s_y)}^{\sqrt{A}(b^{(2)}_k - B s_y)} e^{-u^2} \, du=\frac{\sqrt{\pi}}{2\sqrt{A}} \left[\operatorname{erf}\!\left(\sqrt{A}(b^{(2)}_k \!- B s_y)\right) \!- \operatorname{erf}\!\left(\sqrt{A}(b^{(1)}_k \!- B s_y)\right) \right],
\end{eqnarray*}
where we used the definition of the $\operatorname{erf}(\cdot)$-function\cite{prudnikov1986integrals}. 
Finally, 
\begin{align}
    &P_{c|s_k}=\frac{1}{4} \int_0^\infty \left[ \operatorname{erf}\left(\sqrt{A}(b^{(2)}_k - Bs_y)\right) - \operatorname{erf}\left(\sqrt{A}(b^{(1)}_k - Bs_y)\right) \right]\left[ \operatorname{erf}\!\left(\!\sqrt{A}(b^{(4)}_k - Bs_x)\right) \hspace{-2pt} - \hspace{-1.5pt} \operatorname{erf}\!\left(\!\sqrt{A}(b^{(3)}_k - Bs_x)\right) \right] \hspace{-2pt} f_{|h|^2}(t)dt.\label{P_C_M_QAM}
\end{align}
As an example, for 4-QAM, equation~\eqref{P_C_M_QAM} simplifies to
\begin{equation}
P_{c|s_1} = \int_0^\infty\!\! \left(\! 1 \!-\! Q\left( \sqrt{\frac{s_x^2 t}{\sigma_n^2}}\right)\! \right)\!\left( \!1\! -\! Q\left(\! \sqrt{\frac{s_y^2 t}{\sigma_n^2}}\right)\!\right) f_{|h|^2}(t)\,dt, \label{MMSE_SER}
\end{equation}
which implies,
\begin{align}
P_{e|s_1} &=1-\mathbb{E}_{|h|}\left[ \left(1 - Q\!\left(\Upsilon_x\right)\right) \left(1 - Q\!\left(\Upsilon_y\right)\right)\right]=\underbrace{\mathbb{E}_{|h|}\left[ Q\!\left(\Upsilon_x\right)\right]+\mathbb{E}_{|h|}\left[ Q\!\left(\Upsilon_y\right)\right]}_{{\mathcal{K}}}+\underbrace{\mathbb{E}_{|h|}\left[ Q\!\left(\Upsilon_x\right)Q\!\left(\Upsilon_y\right)\right]}_{\mathcal{S}},\label{MMSEApsyx}
\end{align}
where \(\Upsilon_x \triangleq \frac{ s_x\,|h|}{\sigma_n}\) and \(\Upsilon_y \triangleq \frac{ s_y\,|h|}{\sigma_n}\).


\subsection{Asymptotic Analysis}
\label{sec:asympanalysis}
In order to gain further insights when the SER expressions are written in integral forms and cannot be simplified, we perform an asymptotic analysis at high SNR. This is standard in wireless communications performance analysis since it leads to simplifying the expressions and to capturing the functional relationship between SER and SNR in the high SNR regime. In this context, the approach presented in~\cite{tse2005fundamentals} is adopted. At high SNR $\Upsilon_x$ and $\Upsilon_y$ (that is, for small \(\sigma_n\)), the dominant source of errors arises when the overall channel gain becomes small, which is termed a \textit{deep-fade event}. Starting from \eqref{MMSEApsyx}, the term $\mathcal{S}$ can be neglected, since it involves the factor $ Q\!\left(\Upsilon_x\right)\times Q\!\left(\Upsilon_y\right)$, which remains sufficiently small in the regime considered. Consequently, only the term $\mathcal{K}$ remains, and by applying the approach outlined in~\cite{tse2005fundamentals}, we obtain:
\begin{align*}
P_{e|s_1} &\approx \mathbb{P}(\Upsilon_x < 1)+\mathbb{P}(\Upsilon_y < 1)=\mathbb{P}\!\left(\!|h| \!\!<\!\! \frac{\sigma_n}{s_x}\!\right)\!+\! \mathbb{P}\!\left(\!|h|\!\! < \!\!\frac{\sigma_n}{s_y}\!\right)\!=\! F_{|h|}\!\left(\frac{\sigma_n}{s_x}\right)\!+\!F_{|h|}\!\left(\frac{\sigma_n}{s_y}\right)\!.
\end{align*}
For example, for standard $4$-QAM constellations, we have \(s_x = s_y\), and 
\begin{equation}
P_{e|s_1} \approx 2F_{|h|}\!\left(\frac{\sigma_n}{s_x}\right) \label{eq:deepfade}.
\end{equation}

\subsection{Post-Filtering Detection}\label{sec:Post_filtering}
After linear equalization, symbol detection is performed using a maximum likelihood (ML) detector
\begin{equation}
\frac{P(r|H_m)}{P(r|H_n)} \gtrless_{H_n}^{H_m} 1, \label{ML_rule}
\end{equation}
where $H_m$ and $H_n$ correspond to the transmission of symbols $s_m = (s_{x_m}, s_{y_m})$ and $s_n = (s_{x_n}, s_{y_n})$, respectively. Although determining the optimal likelihood-ratio thresholds can be challenging in non-Gaussian noise environments~\cite{Kassam1988}, the post-filtered received signal is Gaussian under perfect channel state information (CSI).  
For a given channel gain $h$, the effective noise variances are  
\(
\tilde{\sigma}^2_{\mathrm{ZF}} = \frac{\sigma_n^2}{a}, 
\quad 
\tilde{\sigma}^2_{\mathrm{MMSE}} = \frac{a}{b^2}\sigma_n^2,
\) 
with $a = |h|^2$ and $b = |h|^2 + \sigma_n^2$, as implied by equations~\eqref{effective_zf} and~\eqref{MMSE_filter}. The post-filtered symbols are  
\(
s_{\mathrm{ZF}} = (s_x, s_y), 
\quad 
s_{\mathrm{MMSE}} = \left( \frac{a}{b}s_x, \frac{a}{b}s_y \right).
\)
The conditional distributions $P(r|H_m)$ and $P(r|H_n)$ are bivariate Gaussian with means $s_m$ and $s_n$, and covariance matrices  
\(
K_{\mathrm{ZF}} = \tilde{\sigma}^2_{\mathrm{ZF}}\mathbb{I}_2, 
\quad 
K_{\mathrm{MMSE}} = \tilde{\sigma}^2_{\mathrm{MMSE}}\mathbb{I}_2,
\) 
where $\mathbb{I}_2$ is the $2\times2$ identity matrix.  
Assuming an $M$-QAM constellation with equally likely symbols ($P(H_m)=1/M$), the ML decision boundaries between any two symbols $(s_m, s_n)$, $m\neq n$, follow from~\eqref{ML_rule}. For the ZF detector, this simplifies to
\[
2(s_{x_m} \!\!-\! s_{x_n})r_x \!+\! 2(s_{y_m} \!- \!s_{y_n})r_y\! + \!\underbrace{(s_{x_n}^2 \!\!\!+ \!s_{y_n}^2\!\! -\!\! s_{x_m}^2 \!\!\!- s_{y_m}^2\!)}_{\mathcal{Q}}\!\! \gtrless_{H_n}^{H_m} \!0.
\]
A similar expression holds for the MMSE detector after scaling the constellation points by $a/b$.
\[
2(s_{x_m} \!\!-\! s_{x_n})r_x \!+\! 2(s_{y_m} \!- \!s_{y_n})r_y\! +\frac{a}{b} \!\underbrace{(s_{x_n}^2 \!\!\!+ s_{y_n}^2\!\! -\! s_{x_m}^2 \!\!-\! s_{y_m}^2\!)}_{\mathcal{Q}}\!\! \gtrless_{H_n}^{H_m} \!0.\]
The decision regions of the ZF and the MMSE filters are in general non-identical. Although they are similar for standard $4$-QAM constellations, 
where due to the constellation symmetry 
\(|s_{x_m}| = |s_{x_n}| = |s_{y_m}| = |s_{y_n}|\) (differing only in sign), the term $\mathcal{Q}$ cancels out. However, for higher-order M-QAM constellations such as 16-QAM and 64-QAM, for example, the symmetry argument is no longer valid. Consequently, $\mathcal{Q}$ does not cancel out, and the MMSE and ZF detectors yield different decision boundaries. Said differently, using the standard decision boundaries $D_k$, the MMSE receiver will result in a performance that is not equivalent to that of the ZF. This observation is commonly accepted, but to the best of our knowledge, has not been formally stated in the literature.
\section{SER Analysis for Gaussian Channels}\label{sec:SER_gaussian}
In this section, we derive the SER for MMSE detectors under a Gaussian channel, based on the system model defined in~\eqref{generic_sys}.
We reuse the result of~\eqref{MMSE_SER} 
after applying the appropriate modifications. That is, the noise variance becomes $2\sigma_n^2(1 - \lvert \lambda \rvert^2)$ instead of $2\sigma_n^2$ and the transmitted symbols become \(s = \bigl(s_x + \sqrt{2}\lambda_x\sigma_n\bigr) + i\bigl(s_y - \sqrt{2}\lambda_y\sigma_n\bigr)
\) instead of $s = s_x + i s_y$. Substituting $f_{|h|^2}(t)$ by its expression given in~\eqref{PDF_|h|2_GS}, equation~\eqref{MMSE_SER} boils down to 
\begin{equation}
    P_{c|s_1} \!= \! \int_0^\infty \!\!e^{-t} \left(1 - Q\left(a_x\sqrt{t}\right)\right)\left(1 - Q\left(a_y\sqrt{t}\right)\right) dt \label{MMSE_final_PC_GAS_corr},
    \end{equation}
    where \(a_y = \frac{(s_y - \sqrt{2}\sigma_n\lambda_y)}{\sigma_n\sqrt{1-\abs{\lambda}^2}}\), and  \(a_x = \frac{(s_x + \sqrt{2}\sigma_n\lambda_x)}{\sigma_n\sqrt{1-\abs{\lambda}^2}}\). Equation~\eqref{MMSE_final_PC_GAS_corr} can be found in closed-form in the independent case ($\lambda = 0$), and whenever  $s_x=s_y$,
\begin{align}
&P_{c|s_1} = -\underbrace{\int_0^\infty e^{-t} \operatorname{erfc}\left(\sqrt{\frac{s_x^2\,t}{2\sigma_n^2}}\right) \, dt}_{I_1}  + \int_0^\infty e^{-t} \, dt + \frac{1}{4} \underbrace{\int_0^\infty e^{-t} \operatorname{erfc}^2\left(\sqrt{\frac{s_x^2\,t}{2\sigma_n^2}}\right) \, dt}_{I_2}, \notag\\
    &= \frac{1}{4} + \frac{1}{\sqrt{1 + 2\sigma_n^2/s_x^2\,}}-\frac{\arctan\left(\sqrt{1 + 2\sigma_n^2/s_x^2\,}\right)}{\pi\sqrt{1 + 2\sigma_n^2/s_x^2\,}},\label{MMSE_final_PC_GAS_unc}
\end{align}
where we expressed the $Q(\cdot)$-function in terms of the complementary error function $\operatorname{erfc}(\cdot)$. The integral $I_1$ can be computed as follows:
\begin{align}
I_{1} &= \int_{0}^{\infty} e^{-t}\,\mathrm{erfc}\left(\sqrt{\frac{s^2_xt}{2 \sigma_n^2}}\right)dt\nonumber\\
&= \frac{4 \sigma_n^2}{s_x^2}\int_{0}^{\infty} u\,e^{-\frac{2 \sigma_n^2}{s_x^2} u^{2}}\,\mathrm{erfc}(u)\,du \nonumber\\
&= \left(1-\frac{1}{\sqrt{1+ 2 \sigma_n^2/s_x^2}}\right),\label{Appremov}
\end{align}
where equation~\eqref{Appremov} is due to~\cite[2.8.5.9]{prudnikov1986integrals}.
For $I_2$, we use~\cite[8.258.2]{zwillinger2014table}, and equation~\eqref{MMSE_final_PC_GAS_unc} follows.
\paragraph*{Approximation analysis}
For general $\lambda$, we derive an approximate expression of equation~\eqref{MMSE_final_PC_GAS_corr}. Rewriting~\eqref{MMSE_final_PC_GAS_corr} in terms of $\text{erfc}(\cdot)$, we obtain
\begin{align}
P_{c|s_1} &=\! 1 \!-\! \frac{1}{2}\! \underbrace{\int_0^\infty \!\!\!e^{-t} \text{erfc}\!\left(\frac{a_y \sqrt{t}}{\sqrt{2}}\right) \! dt}_{I_3} \!-\! \frac{1}{2} \underbrace{\int_0^\infty \!\!e^{-t} \text{erfc}\left(\frac{a_x \sqrt{t}}{\sqrt{2}}\right) \! dt}_{I_4} + \frac{1}{4} \underbrace{\int_0^\infty e^{-t} \text{erfc}\left(\frac{a_y \sqrt{t}}{\sqrt{2}}\right) \text{erfc}\left(\frac{a_x \sqrt{t}}{\sqrt{2}}\right) dt.}_{I_5} \notag\\
    &\approx \frac{a_y}{2\sqrt{2+a_y^2}}+\frac{a_x}{2\sqrt{2+a_x^2}}+\frac{1}{4}\left[\frac{1}{18(2 + a_x^2 + a_y^2)} +  \frac{1/2}{6 + 3a_x^2 + 4a_y^2}\!+\!\frac{1/2}{6 + 3a_y^2 + 4a_x^2}\!+\!\frac{3/4}{3 + 2a_x^2 + 2a_y^2}\!\right],\label{MMSE_final_PC_GAS_corr_app}
\end{align}
where in order to write equation~\eqref{MMSE_final_PC_GAS_corr_app}, we used the result from~\eqref{Appremov} to evaluate $I_3$ and $I_4$ and we approximate $I_5$ using \(
\text{erfc}(x)\approx \frac{1}{6} e^{-x^2} +\frac{1}{2} e^{-\frac{4x^2}{3}}
\)~\cite{1210748}. 

Furthermore, as presented in Section~\ref{sec:asympanalysis}, the conditional error probability can be approximated asymptotically as  
\begin{align}
    P_{e|s_1} &\approx F_{|h|}\!\left(\frac{\sigma_n}{ s_x+\sqrt{2}\sigma_n\lambda_x}\right)+F_{|h|}\!\left(\frac{\sigma_n}{s_y-\sqrt{2}\sigma_n\lambda_y}\right)\notag
    \\&= 2 - \exp\left(-\frac{\sigma_n^2(1-\abs{\lambda}^2)}{2\, (s_x+\sqrt{2}\sigma_n\lambda_x)^2}\right)-\exp\left(-\frac{\sigma_n^2(1-\abs{\lambda}^2)}{2\, (s_y-\sqrt{2}\sigma_n\lambda_y)^2}\right).\label{eq:asym_gaus_corr}
\end{align}

\section{SER Analysis for THz Channels}
\label{sec:SER_THz}
In this section, we derive the SER under ZF and MMSE detectors for indoor and outdoor THz channels using the framework proposed in~\cite{11072420,11078009} for the ZF detector under Rayleigh and $\alpha$-$\mu$ channels. We consider two cases based on the model in equation~\eqref{generic_Thz_sys}: (i) the independent channel and noise assumption, (ii) the copula-based correlation case. In addition, we consider (iii) the correlated system model defined in~\eqref{Sp_Thz_sys}. For the MMSE detector, we analyze the independent scenario in~\eqref{generic_Thz_sys} and the correlated system model in~\eqref{Sp_Thz_sys}.

\subsection{ZF Detector}
In this subsection, we focus on the ZF detector for the outdoor scenario modeled by the MG distribution. The corresponding indoor scenario based on the $\alpha$-$\mu$ fading model has been thoroughly investigated in our previous work~\cite{11072420}. In addition, the Rayleigh fading case has already been studied in~\cite{11078009}. Therefore, these scenarios are omitted here for brevity and to avoid redundancy.
\subsubsection{Independent Channel and Noise}
\label{ind_ZF}
We compute the SER by analyzing the statistics of the ratio between the additive noise and the channel in~\eqref{generic_Thz_sys}. 
We express the effective noise $z$~\eqref{effective_zf} in polar form as
\begin{equation}
    z=\frac{n}{h} =\frac{|n|}{|h|} e^{j (\theta_n-\theta_h)}  =\uprho e^{j \theta_z} =z_x+jz_y,\label{ratio}
\end{equation}
where $\abs{n}$ and $\abs{h}$ are Rayleigh~\eqref{PDF_n_h_ray} and MG~\eqref{PDF_|h|_MG} distributed, respectively. The phase terms, \( \theta_n \) and \( \theta_h \), and their difference, \(\theta_z\), are all uniformly distributed over \([0, 2\pi]\)~\cite{pewsey2013circular}. The PDF of the ratio of the magnitudes, \( \uprho \), denoted as \( f_\uprho(\rho) \) is
\begin{equation}
    f_\uprho(\rho) = \! \int_{0}^{+\infty} y f_{|n|,|h|}(\rho
    y,y) dy = \!\int_{0}^{+\infty} \!y f_{|n|}(\rho y) f_{|h|}(y) \, dy.\label{temp2}
\end{equation}
Substituting~\eqref{PDF_|h|_MG} in~\eqref{temp2}, 
\begin{equation}
    f_\uprho(\rho) = \frac{\rho}{{\sigma_n^2}} \sum_{i=1}^{K}  \frac{\alpha_i}{\nu^{\beta_i}}\!\int_{0}^{\infty}  \! y^{\beta_i+1} \exp\!\left( -\frac{\rho^2 y^2}{2\sigma_n^2} - \frac{\zeta_i}{\nu} y \right) dy,\label{temp3}
\end{equation}
which can be rewritten in terms of the Meijer G-function as
\begin{align}
f_\uprho(\rho) &= \frac{\rho}{{\sigma_n^2}}\sum_{i=1}^{K} \frac{\alpha_i}{\nu^{\beta_i}} \int_{0}^{+\infty} y^{\beta_i+1} \times G_{0,1}^{1,0} \left( \frac{\rho^2}{2\sigma_n^2}y^2 \middle| \begin{array}{c} - \\ 0 \end{array} \right)   G_{0,1}^{1,0} \left(  \frac{\zeta_i}{\nu} y \middle| \begin{array}{c} - \\ 0 \end{array} \right) \, dy. \notag\\
&= \frac{\rho\nu^2}{{\sqrt{2\pi}}\sigma_n^2}\sum_{i=1}^{K} \alpha_i \frac{2^{\beta_i+\frac{3}{2}}}{\zeta_i^{\beta_i+2}}G_{2,1}^{1,2} \left( \frac{2\rho^2\nu^2}{\sigma_n^2\zeta_i^2} \middle| \begin{array}{c} \frac{-1-\beta_i}{2} , \frac{-\beta_i}{2}\\ 0 \end{array}\!\!\right), \notag
\end{align}
where the final equality follows from~\cite{wolfram}. 
Furthermore, because \( \rho \) and \( \theta_z \) are independent, the PDF of the ratio \( z \) is expressed as 
\begin{equation*}
f_{(\uprho,\Theta)}(\rho,\theta) \!=\! \frac{\rho\nu^2}{{{\pi}^\frac{3}{2}}\sigma_n^2} \sum_{i=1}^{K} \!\alpha_i \frac{2^{\beta_i}}{\zeta_i^{\beta_i+2}} G_{2,1}^{1,2} \left(\! \frac{2\rho^2\nu^2}{\sigma_n^2\zeta_i^2} \middle|\!\!\!\! \begin{array}{c} \frac{-1-\beta_i}{2} , \frac{-\beta_i}{2}\\ 0 \end{array} \!\!\!\!\right)\!, 
\end{equation*}
which gives in cartesian coordinates 
\begin{align}
f_{Z_y,Z_x}(z_y, z_x)&= \frac{\nu^2}{{({\pi})^\frac{3}{2}}\sigma_n^2}\sum_{i=1}^{K}  \frac{\alpha_i 2^{\beta_i}}{\zeta_i^{\beta_i+2}}  G_{2,1}^{1,2} \left( \frac{2(z_x^2 + z_y^2)\nu^2}{\sigma_n^2\zeta_i^2} \middle| \begin{array}{c} \frac{-1-\beta_i}{2} , \frac{-\beta_i}{2}\\ 0 \end{array} \right). \label{ZF_final_fz_MG_unc}
\end{align}
Finally, the SER is evaluated by substituting~\eqref{ZF_final_fz_MG_unc} into~\eqref{SER-ZF}.

\subsubsection{Copula Model}
\label{Copula_THz}
We model the correlation between \(h\) and \(n\) using the copula method~\eqref{copula_PDF}. The copula term,\\ \( C(F_{|n|}(\sqrt{z_r^2+z_i^2}\, y), F_{|h|}(y)) \), is evaluated based on~\eqref{FGM} and~\eqref{FRANK} for the FGM and Frank copulas, respectively. The SER \( P_e \) follows from~\eqref{SER-ZF}.
For the outdoor case, the channel magnitude \(|h|\) is MG distributed according to~\eqref{PDF_|h|_MG}. 
The PDF of $z$ is
\begin{align}
    &f_{\Re(z), \Im(z)}(z_r, z_i)\!=\!\sum_{i=1}^{K} \frac{\alpha_i \nu^{-\beta_i}}{2\pi\sigma_{n}^2}\! \int_0^\infty \! y^{\beta_i + 1}\exp\left(\!-\frac{(z_i^2 + z_r^2) y^2}{2\sigma_{n}^2} - \zeta_i y \!\right)
    C(F_{|n|}(\sqrt{z_r^2+z_i^2}\, y), F_{|h|}(y)) dy, \label{MG_copula}
\end{align}
where $F_{|h|}(\cdot)$ and $ F_{|n|}(\cdot)$ are given in~\eqref{CDF_|h|_MG} and~\eqref{PDF_n_h_ray}, respectively.

\subsubsection{Specific THz model}
We address the system model in~\eqref{Sp_Thz_sys}, which incorporates distortion noise from THz-band hardware impairments, exhibiting channel correlation at the receiver~\cite{8610080}. This model can be extended to account for channel-induced molecular absorption noise~\cite{Sarieddeen9514889}. 
For the outdoor case, the analysis for obtaining the PDF of the effective noise in~\eqref{Sp_Thz_sys} is provided in Appendix~\ref{ZF_PC_MG}, where we also determine the probability of error \( P_e \) as follows:
\begin{align}
    &P_{e|s_1}\!=1\!-\!\sum_{i=1}^{K}\!\frac{\alpha_i}{\nu^{\beta_i}} \int_{0}^{\infty}\hspace{-2.5mm}y^{\beta_i+1}\!\left(\!1\!-\!Q\!\left(\frac{s_x}{\sqrt{\bar{\sigma}_{tr}^2\!+\!\frac{\sigma_n^2}{y^2}}}\!\right)\!\right)^2\!\!\exp\left(\!\!-\frac{\zeta_i}{\nu}y\right)dy.\label{ZF_PE_MG_SPsys}
\end{align}
We further analyze the asymptotic behavior of $P_{e|s_1}$ at high SNR. Let \( \bar{\Upsilon} \triangleq \frac{s_x h}{\sqrt{\bar{\sigma}_{tr}^2 h^2 + \sigma_n^2}} \), then, equation~\eqref{ZF_PE_MG_SPsys} becomes
\begin{align}
    P_{e|s_1} &= \!1-\mathbb{E}_{|h|}\left[\left(1\!-\!Q \left(\bar{\Upsilon} \right)\right)^2\right]= 2\mathbb{E}_{|h|}\left[Q \left(\bar{\Upsilon}\right) \right]-\mathbb{E}_{|h|}\left[Q^2 \left(\bar{\Upsilon}\right)\right],\notag
\end{align}
where the PDF, \( f_{|h|}(y) \), is given in~\eqref{PDF_|h|_AM}. We analyze the probability of a ``deep-fade" event. At high SNR (small $\sigma_n$), equation~\eqref{eq:deepfade} implies
\begin{equation}
    P_{e|s_1}  \approx 2\mathbb{P}(\bar{\Upsilon} < 1) =  2F_{|h|}\left(\frac{\sigma_n}{\sqrt{s_x^2-\bar{\sigma}_{tr}^2}}\right),\label{eq:out_Asym}
\end{equation}
where $F_{|h|}(.)$ is defined in~\eqref{CDF_|h|_MG}. For small $\sigma$, the CDF argument is small and we approximate it as \(F_{|h|}(x)\approx \sum_{i=1}^{K} \frac{\alpha_i \zeta_i^{\beta_i}}{\beta_i \nu^{\beta_i}} y^{\beta_i},\) resulting in an asymptotic expression $P_e \!\propto\! \left( \frac{s_x^2 - \bar{\sigma}_{tr}^2}{\sigma_n^2} \right)^{-\frac{\beta_i}{2}}$, where $\frac{s_x^2 - \bar{\sigma}_{tr}^2}{\sigma_n^2}$ represents the ratio between the symbol minus channel-dependent noise powers and the channel-independent noise power. The asymptotic slope, $\frac{\beta_i}{2}$, matches that of~\cite{11362947,10663782} for the SISO case with independent channel and noise.

\subsection{MMSE Detector}
We will examine two scenarios: the first assumes an independent channel and noise, modeled using~\eqref{generic_Thz_sys}, while the second considers the dependency case, represented by the system model in~\eqref{Sp_Thz_sys}.
\subsubsection{Independent Channel and Noise}
We use equation~\eqref{MMSE_SER} for \(P_{{c|s_1}}\), 
where \( f_{|h|^2}(t) \) is given by~\eqref{PDF_|h|2_MG} for the outdoor scenario and by~\eqref{PDF_|h|2_AM} for the indoor scenario. Next, we find approximations and asymptotic expressions.
\paragraph{Outdoor THz-band channels}
We have
\begin{align}
    P_{c|s_1} &= \sum_{i=1}^{K} \frac{\alpha_i}{2 \nu^{\beta_i}} \int_0^\infty t^{\frac{\beta_i}{2}-1} e^{-\frac{\zeta_i}{\nu} \sqrt{t}} \left(\!\! 1 -\! Q\left( \sqrt{\frac{s_x^2 t}{\sigma_{n}^2}}\! \right)\!\!\! \right)\left( \!1\! -\! Q\left(\!\! \sqrt{\frac{s_y^2 t}{\sigma_{n}^2}}\!\right)\!\!\right) dt. \label{MMSE_MG_PC_unc} 
\end{align}
Using Lemma~\ref{Theorem1} in Appendix~\ref{app:integral}, equation~\eqref{MMSE_MG_PC_unc} can be approximated as
\begin{align}
&P_{c|s_1} \approx \sum_{i=1}^{K} \alpha_i \left[ \frac{\Gamma\left(\beta_i\right)}{\zeta_i^{\beta_i}} - \frac{1}{\nu^{\beta_i}} \left( \Phi\left(\beta_i, \frac{\zeta_i}{\nu}, 1, \frac{s_x}{\sqrt{2}\sigma_n}\right) + \Phi\left(\beta_i, \frac{\zeta_i}{\nu}, 1, \frac{s_y}{\sqrt{2}\sigma_n}\right) \right) +\frac{1}{8 \nu^{\beta_i}} \sum_{j=1}^{4} \Psi(D_j, C_j) \right], \label{MMSE_MG_PC_unc_app}
\end{align}
where $\Psi(D_j, C_j)$ and $\Phi(\cdot,\cdot,\cdot,\cdot)$ are defined in Lemma~\ref{Theorem1}. Furthermore, as shown in Section~\ref{sec:asympanalysis}, the asymptotic expression is
\begin{equation}
    P_{e|s_1} \approx 
    2F_{|h|}\!\left(\frac{\sigma_n}{s_x}\right)
    = 2\sum_{i=1}^{K} \alpha_i\, \zeta_i^{-\beta_i}\,
    \gamma\!\left(\beta_i,\,
    \frac{\zeta_i}{\nu}\,
    \frac{\sigma_n}{s_x}\right),\label{Out_MMSE_asym}
\end{equation}
where \(F_{|h|}(\cdot)\) is expressed in~\eqref{CDF_|h|_MG}.
\paragraph{Indoor THz-band channels}
We have 
\begin{align}
P_{c|s_1} = &\varphi\int_0^\infty t^{\frac{\alpha \mu }{2}-1} \exp\left(-\frac{\mu t^{\alpha/2}}{(\nu \hat{Z})^\alpha} \right) \left(\!1 \!-\! Q\left( \sqrt{\frac{s_x^2 t}{\sigma_n^2}} \right)\!\right)\!\left( \!1\! -\! Q\left(\! \sqrt{\frac{s_y^2 t}{\sigma_n^2}}\right)\!\right) dt, \label{MMSE_AM_PC_unc}
\end{align}
where \(\varphi \triangleq \frac{\alpha \mu^\mu}{2 (\hat{Z}\nu)^{\alpha \mu}\Gamma(\mu)}\). Again, using Lemma~\ref{Theorem1}, we get the following approximation:
\begin{align}
&P_{c|s_1} \approx 1-\varphi \left[  - \Phi\left(\alpha \mu, \frac{\mu}{(\nu \hat{Z})^\alpha}, \alpha, \frac{s_x}{\sqrt{2}\sigma_n}\right) - \Phi\left(\alpha \mu, \frac{\mu}{(\nu \hat{Z})^\alpha}, \alpha, \frac{s_y}{\sqrt{2}\sigma_n}\right) \!+\! \frac{1}{4} \sum_{j=1}^{4} \Psi(D_j, C_j) \right]. \label{MMSE_AM_PC_unc_app}
\end{align}
At high SNR, the asymptotic expression given by equation~\eqref{eq:deepfade} boils down to 
\begin{equation}
    P_{e|s_1} \approx 2\,F_{|h|}\!\left(\frac{\sigma_n}{s_x}\right)
    = 2\,\frac{\gamma\!\left(\mu,\, \mu\left(\frac{\sigma_n}{s_x\, \hat{Z}\, \nu}\right)^{\!\alpha}\right)}{\Gamma(\mu)}.\label{eq:asym_indoor_4Qam}
\end{equation}
\subsubsection{Correlated Channel and Noise}\label{MMSE_ZF_corr}
Using the MMSE detector, the signal received under the channel model of~\eqref{Sp_Thz_sys} is similar to that given by~\eqref{MMSE_filter}, however, the noise variance is modified to include both thermal and distortion noises as follows:
\begin{equation}
    \tilde{r} = \frac{h^* \sigma_s^2}{|h|^2 \sigma_s^2 + 2\sigma_{trn}^2} r, \quad 2\sigma_{trn}^2 = 2\bar{\sigma}_{tr} |h|^2 + 2\sigma_n^2. \label{MMSE_filter_THz}
\end{equation}
\paragraph{Outdoor THz-band channels}
Using equations~\eqref{MMSE_SER} and~\eqref{PDF_|h|2_MG}, we obtain
\begin{align}
P_{c|s_1} &=\sum_{i=1}^{K}  \frac{\alpha_i}{2 \nu^{\beta_i}} \int_0^\infty t^{\frac{\beta_i }{2}-1}  \exp\left({-\frac{\zeta_i}{\nu} \sqrt{t}}\right)\left(\!\! 1 -\! Q\left( \sqrt{\frac{s_x^2 t}{\bar{\sigma}_{tr} t+\sigma_n^2}}\! \right)\!\!\! \right)\!\!\!\left( \!1\! -\! Q\left(\!\! \sqrt{\frac{s_y^2 t}{\bar{\sigma}_{tr} t+\sigma_n^2}}\!\right)\!\!\right)\ dt.\label{MMSE_MG_PC_corr}
\end{align}
Equation~\eqref{MMSE_MG_PC_corr} cannot be solved directly using Lemma~\ref{Theorem1}. However, one can still find the asymptotic expression at high SNR ($\sigma_n\rightarrow0$): 
\begin{eqnarray}
    P_{e|s_1} &\approx& 2F_{|h|}\left(\!\frac{\sigma_n}{\sqrt{s_x^2 - \bar{\sigma}_{tr}}}\!\right) = 2\sum_{i=1}^{K} \alpha_i\, \zeta_i^{-\beta_i}\,
    \gamma\!\left(\beta_i,
    \frac{\zeta_i\sigma_n}{\nu\sqrt{s_x^2 - \bar{\sigma}_{tr}}}\right),
\end{eqnarray}  
where \(F_{|h|}(\cdot)\) is given in~\eqref{CDF_|h|_MG}.
\paragraph{Indoor THz-band channels}
Using equations~\eqref{PDF_|h|2_AM} and~\eqref{MMSE_SER}, 
\begin{align}
P_{c|s} & = \varphi \int_0^\infty t^{\frac{\alpha \mu }{2}-1} \exp\left(-\frac{\mu t^{\alpha/2}}{(\nu \hat{Z})^\alpha} \right) \left(\! 1 -\! Q\left(\!\sqrt{\frac{s_x^2 t}{\bar{\sigma}_{tr} t+\sigma_n^2}}\right)\! \right)\!\left( \!1\! -\! Q\left(\!\sqrt{\frac{s_y^2 t}{\bar{\sigma}_{tr} t+\sigma_n^2}}\right)\!\right) dt. \label{MMSE_AM_PC_corr}
\end{align}
The asymptotic expression at high SNR is expressed as 
\begin{equation}
    P_{e|s_1} \approx 
    2F_{|h|}\!\left(\frac{\sigma_n}{\sqrt{s_x^2 - \bar{\sigma}_{tr}}}\right)
    =2 \frac{\gamma\!\left(\mu,\, \mu \left(\frac{\sigma_n}{\left(\sqrt{s_x^2 - \bar{\sigma}_{tr}}\right)\hat{Z}\nu}\right)^{\!\alpha}\right)}{\Gamma(\mu)}.
\end{equation}
\section{Simulation Results}\label{sec:simultaion}
This section presents simulation results that validate the theoretical analysis. The simulations were conducted using realistic parameters to ensure accuracy and reliability, with the implementation carried out in MATLAB. We note that the simulation results for ZF detection under Gaussian channels and indoor THz fading are based on the analytical framework developed in our prior works~\cite{11072420,11078009}. In addition to these established results, we extend the numerical analysis to outdoor THz environments for the ZF detector, as well as to the MMSE detector for both Gaussian and THz channels in indoor and outdoor scenarios.

\subsection{Gaussian Channel}

The SER of the linear ZF and MMSE detectors is numerically evaluated using Monte Carlo simulations with $N = 5\times10^6$ independent trials. The simulated results are compared with the proposed analytical expressions. The communication channel is modeled under Rayleigh fading, potentially correlated with additive white Gaussian noise (AWGN). Normalized, equiprobable $M$-QAM constellations are considered for $M = 4$ and $M = 16$. Specifically, for the $4$-QAM case, the constellation symbols $s_x, s_y$ take values from $\{\pm1\}$, whereas for the $16$-QAM case, they take values from $\{\pm1, \pm3\}$.

The simulation results are illustrated in Fig.~\ref{fig:Gaussian_4QAM}, where both ZF and MMSE detectors are evaluated under independent and correlated fading scenarios. For 4-QAM modulation, the SER performance of the ZF and MMSE detectors is identical, confirming the theoretical analysis presented in Section~\ref{sec:Post_filtering}. Specifically, for the independent case ($\lambda = 0$), the theoretical SER curves correspond to~\eqref{MMSE_final_PC_GAS_unc} and~\cite[Eq. 8]{11078009} for the MMSE and ZF detectors, respectively. In the correlated case, where $\lambda \in \{\,0.7 + 0.7j,\; 0.65 - 0.6j\,\}$, the theoretical SER expressions are given by~\eqref{MMSE_final_PC_GAS_corr_app} and~\cite[Eq. 8]{11078009} for the MMSE and ZF detectors, respectively, with the asymptotic behavior described by~\eqref{eq:asym_gaus_corr}. The simulation results show excellent agreement with the corresponding theoretical predictions. 
The analysis is further extended to higher-order modulations.
Fig.~\ref{fig:Gaussian_16QAM} shows the SER for 16-QAM. For instance, the ZF presents a gain of \unit[1.13]{dB} at SNR equal \unit[10]{dB} over the MMSE detector. 
This performance difference is due to the mismatch in the MMSE detection regions, as theoretically derived in Section~\ref{sec:Post_filtering}. The average SER is obtained using the effective noise PDF $f_{z_x,z_y}$ given by~\cite[Eq. 7]{11078009} for the ZF detector and by~\eqref{P_C_M_QAM} for the MMSE detector.
In addition, Figures~\ref{fig:Gaussian_4QAM} and~\ref{fig:Gaussian_16QAM} reveal that the SER obtained under the independent fading assumption serves as an upper bound to the actual SER in the correlated case. This highlights the tendency of the independent-case analysis to overestimate the SER when correlation effects are neglected. For example, at $\abs{\lambda} \approx 1$, corresponding to a highly correlated channel, the SER performance deviates significantly from that predicted under the independence assumption~\cite{11072420,11078009}.

\begin{figure}[ht!]
    \centering
    \includegraphics[width=0.5\textwidth]{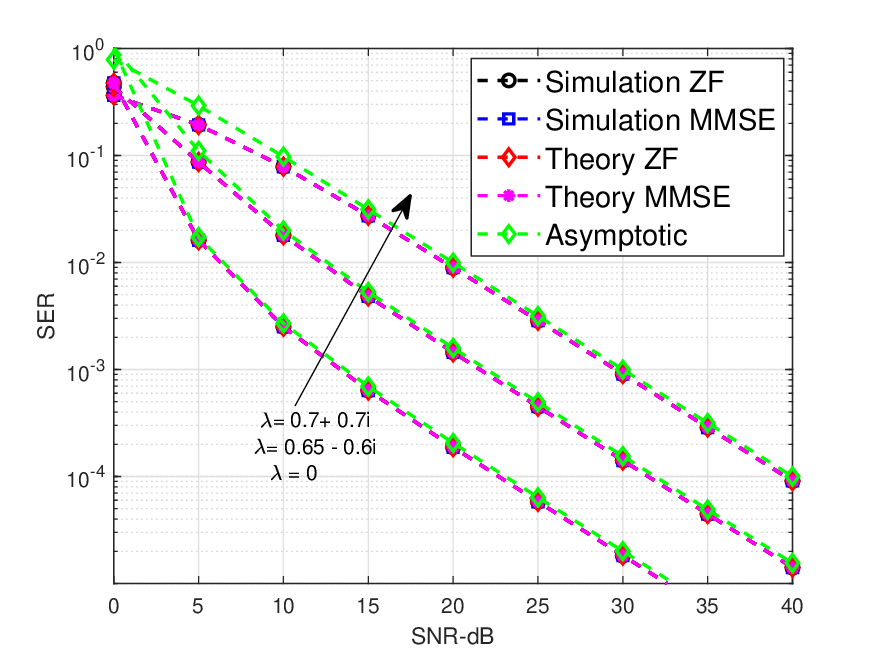}
    \caption{SER vs. SNR for 4-QAM with ZF and MMSE detection under Rayleigh fading.}
    \label{fig:Gaussian_4QAM}
\end{figure}
\begin{figure}[ht!]
    \centering
    \includegraphics[width=0.5\textwidth]{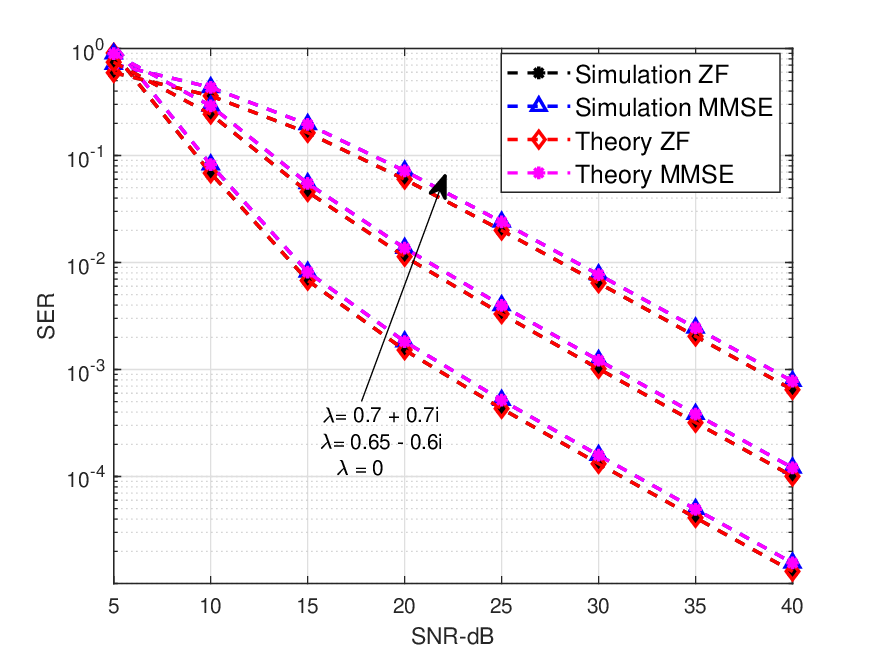}
    \caption{SER vs. SNR for 16-QAM with ZF and MMSE detection under Rayleigh fading.}
    \label{fig:Gaussian_16QAM}
\end{figure}

\subsection{THz Channel}
\label{Performance_Evaluation}

For the THz setup, the noise power is defined as $k_{\mathrm{B}}TB$, where $k_{\mathrm{B}}$ denotes the Boltzmann constant, $T = \unit[300]{K}$ is the absolute temperature, and $B$ represents the system bandwidth. The carrier frequency is set to $f = \unit[0.142]{THz}$ with a bandwidth of $B = \unit[4]{GHz}$. The antenna gains are $G_r = \unit[19]{dBi}$ and $G_t = \unit[0]{dBi}$~\cite{papasotiriou2023outdoor}. The transmit power $p_t$ is varied to span different SNR regimes.

\subsubsection{Indoor THz}
Fast fading is modeled using the $\alpha$--$\mu$ distribution. The simulations assume \{16, 4\}-QAM modulation. The communication distance is varied according to practical parameter values, with $\alpha = 3.20544$, $\mu = 7.80577$, and $\hat{Z} = 14.0391$, as reported in~\cite[Table~I]{11072420}. It is worth noting that part of the presented results are recalled from our previous work~\cite{11072420} for completeness.

Figure~\ref{fig:indoor_4_16QAM1} illustrates the SER versus SNR performance of the MMSE and ZF detectors under $4$-QAM and $16$-QAM modulations. For the $4$-QAM case, the theoretical SER for the MMSE detector, given by~\eqref{MMSE_AM_PC_unc}, along with its approximation in~\eqref{MMSE_AM_PC_unc_app} and the asymptotic expression in~\eqref{eq:asym_indoor_4Qam}, are compared with the simulation results. The theoretical SER for the ZF detector is provided in~\cite[Eq. 13]{11072420}. The strong agreement between theoretical and simulated results validates the accuracy of the proposed SER expressions. As expected and consistent with the analysis in Section~\ref{sec:Post_filtering}, the MMSE and ZF detectors exhibit identical performance for $4$-QAM.
The same figure also presents the SER performance under $16$-QAM. The theoretical SER expression for the MMSE detector is given in~\eqref{P_C_M_QAM}, whereas for the ZF detector, the average SER is obtained using the effective noise PDF $f_{z_x,z_y}$~\cite[Eq. 11]{11072420} and the SER expression in~\eqref{SER-ZF}. Under mismatched decision regions, the MMSE detector performs worse than the ZF detector under $16$-QAM. However, at low SNR, the mismatch is relatively small, \unit[0.95]{dB} at SNR equal \unit[5 dB]{}, which justifies the practical use of MMSE with mismatched regions. This choice not only simplifies implementation but also provides improved numerical stability compared to ZF.
\begin{figure}[t!]
    \centering
    \includegraphics[width=0.5\textwidth]{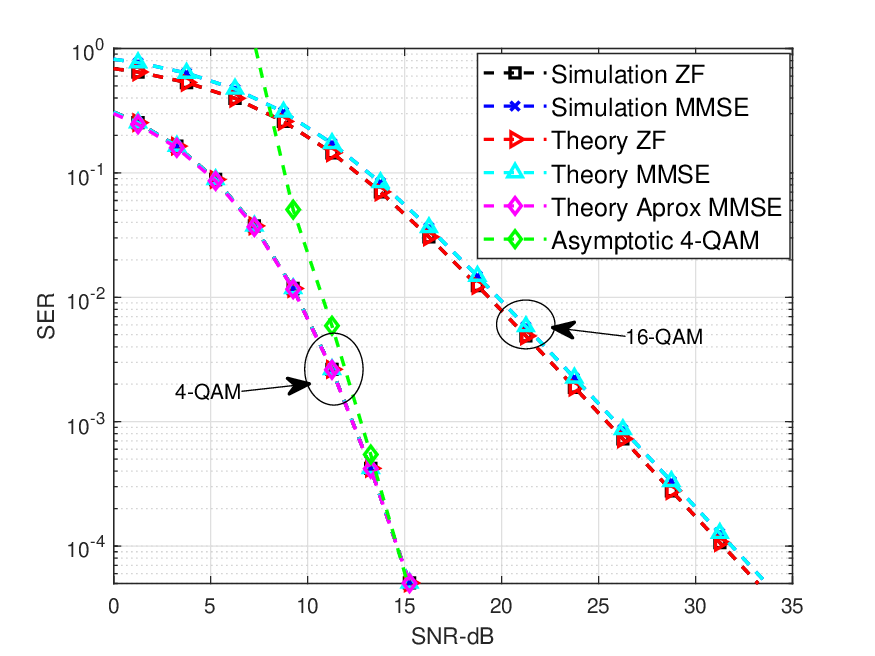}
    \caption{SER vs. SNR for ZF and MMSE detection over indoor THz channels with independent channel and noise.}
    \label{fig:indoor_4_16QAM1}
\end{figure}


\subsubsection{Outdoor THz}
We extend the performance analysis to an outdoor THz channel using an MG fading model. The results presented in figure~\ref{fig:ZF_Out_fgm_ind}, ~\ref{fig:ZF_Out_frank_ind}, and~\ref{fig:ZF_Out_newM} provide a comprehensive analysis of SER performance of ZF detection for outdoor THz channels, under various scenarios.
\begin{figure}[hpbt!]
    \centering
    \includegraphics[width=0.5\textwidth]{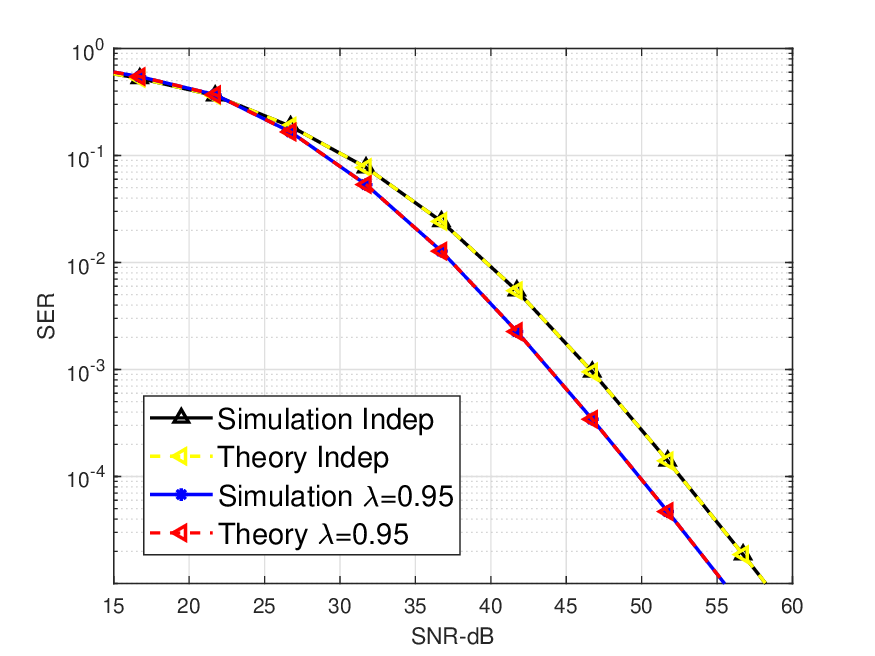}
    \caption{SER vs. SNR for ZF detection over outdoor THz channels with correlated channel and noise using the FGM copula.}
    \label{fig:ZF_Out_fgm_ind}
\end{figure}
\begin{figure}[hpbt!]
    \centering
    \includegraphics[width=0.5\textwidth]{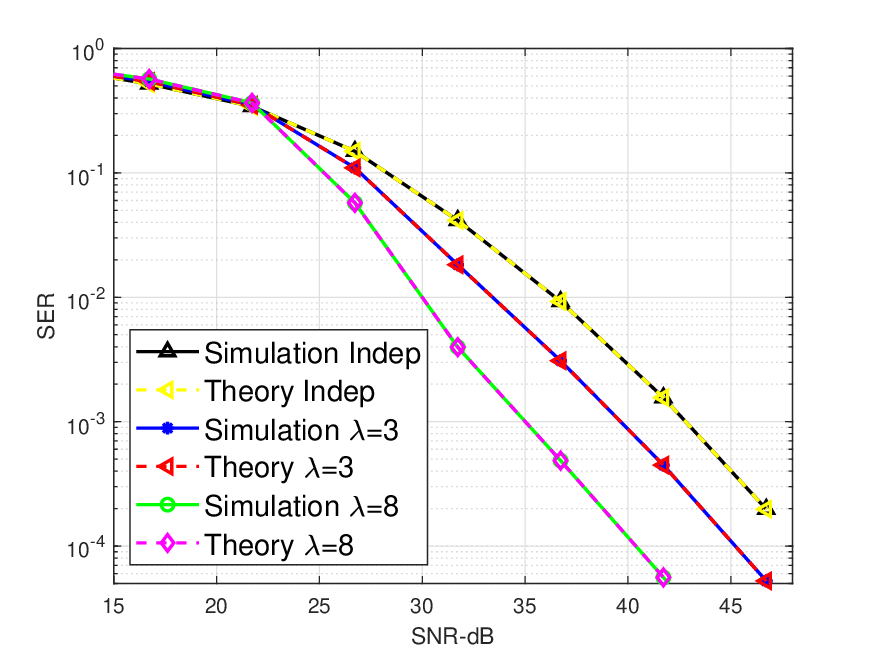}
    \caption{SER vs. SNR for ZF detection over outdoor THz channels with correlated channel and noise using the Frank copula.}
    \label{fig:ZF_Out_frank_ind}
\end{figure}
\begin{figure}[hpbt!]
    \centering
    \includegraphics[width=0.5\textwidth]{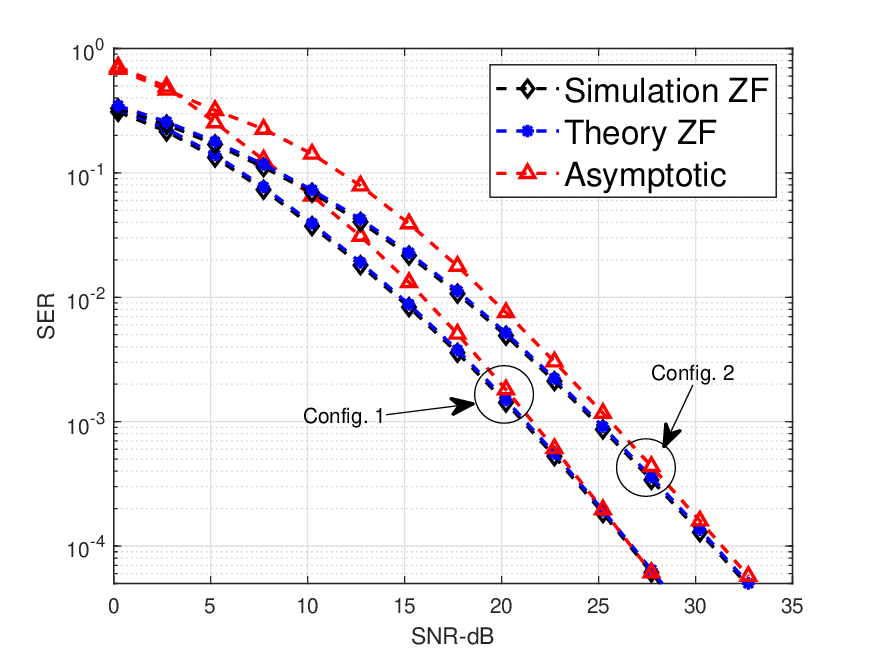}
    \caption{SER vs. SNR for ZF detection over outdoor THz channels with distortion noise.}
    \label{fig:ZF_Out_newM}
\end{figure}

\begin{figure}[t!]
    \centering
    \includegraphics[width=0.5\textwidth]{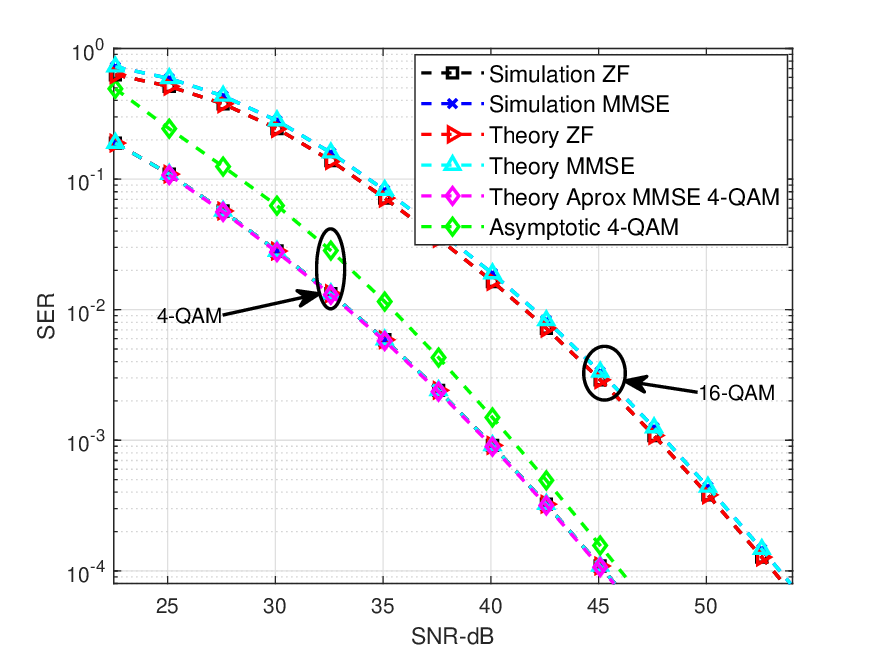}
    \caption{SER for THz outdoor scenarios under independent channel and noise for MMSE and ZF with \{16,4\}-QAM, where MG parameters follow Config 1 from Table~\ref{tab:MG}.}
    \label{fig:MG_MMSE_ZF}
\end{figure}

\begin{table}[!t]
\centering
\caption{Parameter configurations for MG.}
\resizebox{0.5\columnwidth}{!}{%
\begin{tabular}{c c c c c}
\hline
$K$ & $w$ & $\beta$ & $\zeta$ & \textbf{Config} \\
\hline
\multirow{2}{*}{K=2} & 0.67540627 & 15.2709327  & 0.069986341 & \multirow{2}{*}{Config 1\cite{papasotiriou2023outdoor}} \\
 & 0.32459373 & 4.417045104 & 0.153163953 &  \\
\hline
\multirow{2}{*}{K=2} & 0.512500204 & 3.75341946  & 0.159416427 & \multirow{2}{*}{Config 2\cite{papasotiriou2023outdoor}} \\
 & 0.487499796 & 22.59894871 & 0.054461913 &  \\
\hline
\end{tabular}%
}
\label{tab:MG}
\end{table}
Figures~\ref{fig:ZF_Out_fgm_ind} and \ref{fig:ZF_Out_frank_ind} illustrate the SER performance for both independent and correlated cases using the FGM and Frank copula models, respectively, where the MG parameters are selected according to Config~1 in Table~\ref{tab:MG}. The theoretical results are obtained from \eqref{P_C_M_QAM}, in which the average SER is evaluated using the effective noise PDF $f_{z_x,z_y}$ given in \eqref{MG_copula}. For the independent case, the correlation parameter is set to $\lambda=0$. 
Figure~\ref{fig:ZF_Out_newM} illustrates the impact of receiver distortion noise for the two configurations presented in Table~\ref{tab:MG}. The theoretical results are obtained from~\eqref{ZF_PE_MG_SPsys}, while the asymptotic results are derived from~\eqref{eq:out_Asym}, assuming $\kappa_t = \kappa_r = 0.2$. As observed, the asymptotic analysis closely matches the theoretical curves in the high-SNR region, confirming the accuracy of the derived expressions. 

Figure~\ref{fig:MG_MMSE_ZF} extends the simulation analysis to jointly evaluate the MMSE and ZF detectors for both 4-QAM and 16-QAM modulation schemes. For 4-QAM, the theoretical SER curves for the ZF detector are obtained by substituting~\eqref{ZF_final_fz_MG_unc} into~\eqref{SER-ZF}, whereas the MMSE detector results are derived from~\eqref{MMSE_MG_PC_unc}. The corresponding approximation for the MMSE detector follows from~\eqref{MMSE_MG_PC_unc_app}, and the asymptotic behavior is characterized by~\eqref{Out_MMSE_asym}. The theoretical and simulation results show an excellent agreement, with the approximation curves also accurately capturing the SER trends.
The figure further extends the analysis to 16-QAM, where a difference between the ZF and MMSE detector performance is observed, as expected. The average SER for the ZF detector is computed using the effective noise PDF $f_{z_x,z_y}$ from~\eqref{ZF_final_fz_MG_unc} substituted into~\eqref{SER-ZF}, whereas for the MMSE detector the expression in~\eqref{P_C_M_QAM} is used. 

\section{Conclusion}\label{sec:conclusion}
This paper provides a detailed performance analysis of linear detection schemes under channel-noise correlation, namely mismatched MMSE detectors, over Rayleigh, $\alpha$-$\mu$, and mixture of Gamma channels- the latter two being well-known models of indoor and outdoor THz communication scenarios. Moreover, we complement the analysis conducted in~\cite{11072420,11078009} for the ZF detector over Rayleigh and $\alpha$-$\mu$ channels, by considering the mixture of Gamma model.  Semi-analytical, approximate, and asymptotic expressions for the SER were derived while explicitly incorporating channel-noise correlation through both statistical and copula-based modeling approaches. The derived analytical results were shown to closely match the Monte Carlo simulation outcomes, demonstrating the accuracy and robustness of the proposed framework. Furthermore, the analysis highlighted the significant impact of the channel-noise correlation on detector performance, leading to noticeable SER deviations compared to uncorrelated cases. The presented results establish a reliable analytical foundation for evaluating and optimizing linear detectors in future THz communication systems. It is worth noting that the MMSE detector inherently introduces a bias in the detected symbols due to its filtering term. Although this bias reduces noise amplification, it slightly shifts the standard ZF decision boundaries in higher order modulations ($M > 4$).
The analysis of the MMSE detector, when operating under ZF decision regions, is motivated by its robustness to noise amplification, particularly in ill-conditioned channel conditions. We show that the resulting degradation gap remains limited (e.g., approximately 1~dB for 16-QAM at 5~dB SNR), while offering improved robustness in deep fading scenarios where ZF significantly amplifies noise. Overall, these findings demonstrate that MMSE, even under suboptimal decision regions, provides a favorable trade-off between noise enhancement and detection reliability, making it a practical and robust alternative to ZF in realistic THz systems.

\appendices
\counterwithin*{equation}{section}
\renewcommand\theequation{\thesection.\arabic{equation}}
\appendices
\section{Technical Lemma} 
\label{app:integral}

\begin{lemma}
\label{Theorem1}
Let $a > 1$, and let $b$, $c$, $n$, and $\delta$ be positive numbers. Then
\begin{align}
I &= \hspace{-0.5pt} \!\int_0^\infty \! \hspace{-0.5pt} x^{a-1} \exp\left(-b x^{n}\right) \left(1 - Q\left(\sqrt{ cx} \right)\right)\left(1 - Q\left(\sqrt{\delta x} \right)\right) \!dx \notag\\
&\approx \frac{1}{n b^{\frac{a}{n}}} \Gamma\left(\frac{a}{n}\right) - \Phi\left(2a, b, 2n, \frac{\sqrt{c}}{\sqrt{2}}\right) - \Phi\left(2a, b, 2n, \frac{\sqrt{\delta}}{\sqrt{2}}\right) + \frac{1}{4} \sum_{j=1}^{4} \Psi(D_j, C_j),\notag
\end{align}  
where
\begin{eqnarray*}
\textbf{C} = \left[\frac{c + \delta}{2}, \frac{2c + \delta}{3}, \frac{c + 2\delta}{3}, \frac{2(c + \delta)}{3}\right],\qquad
\textbf{D} = [1/36, 1/12, 1/12, 1/4],
\end{eqnarray*}
and for $1 \leq j \leq 4$
\begin{align}
\Psi(D_j,C_j) \!=\!D_j \frac{k^{0.5}l^{a-0.5}}{(2\pi)^{\frac{k+l}{2}-1}C_j^a} G_{l,k}^{k,l}\!\left(\frac{(b/k)^k}{(C_j/l)^l}\middle| \begin{array}{c}I(l,1-a)\\I(k,0)\end{array}\right),\notag
\end{align}
where $n \!=\! \frac{l}{k}$, \(l\) and \(k\) are coprime (ensuring non-integer values of $n$ are accounted for), and $ I(p, q) \!=\! \frac{q}{p}, \frac{q+1}{p}, \ldots, \frac{q +p-1}{p} $.
\end{lemma}
\begin{proof}
We rewrite $I$ in terms of $\text{erfc}(.)$:
\begin{align}
I &=\underbrace{\int_0^\infty\!x^{a-1} e^{-b x^{n}} dx}_{I_1} - \frac{1}{2} \underbrace{\int_0^\infty x^{a-1} e^{-b x^{n}} \text{erfc}\left(\frac{\sqrt{cx}}{\sqrt{2}}\right) dx}_{I_2} - \frac{1}{2} \underbrace{\int_0^\infty x^{a-1} e^{-b x^{n}} \text{erfc}\left(\frac{\sqrt{\delta x}}{\sqrt{2}}\right) dx}_{I_3} \notag\\&+ \frac{1}{4} \underbrace{\int_0^\infty x^{a-1} e^{-b x^{n}} \text{erfc}\left(\frac{\sqrt{cx}}{\sqrt{2}}\right) \text{erfc}\left(\frac{\sqrt{\delta x}}{\sqrt{2}}\right) dx}_{I_4}.    
\end{align}
We first note that \(I_1 =\frac{1}{n b^{\frac{a}{n}}} \Gamma\left(\frac{a}{n}\right)
\)~\cite[3.326.2]{zwillinger2014table}. For $I_2$ and $I_3$, we write them in terms of the function \(\Phi(u, p, r, \omega)\) defined as~\cite[eq. (2.8.1.5)]{prudnikov1986integrals}:
\begin{equation}
\Phi(u, p, r, \omega) = \int_{0}^{\infty} x^{u-1} e^{-p x^r} \operatorname{erfc}(\omega x) \,dx,\label{temp5}
\end{equation}
where, for $r>1$,
\begin{align*}
\Phi (u, p, r, \omega) &= - \frac{\omega}{\sqrt{\pi} p^{(u+1)/r} r} \sum_{k=0}^{\infty} \frac{(-1)^k}{(k + 1/2) k!} 
\Gamma \left( \frac{2k + u + 1}{r} \right) \left( \frac{\omega}{p^{1/r}} \right)^{2k} + \frac{1}{p^{u/r} r} \Gamma \left( \frac{u}{r} \right),
\end{align*}
and for $r\leq1$
\begin{equation*}
\Phi (u, p, r, \omega)= \frac{1}{\omega^u \sqrt{\pi}} \sum_{k=0}^{\infty} \frac{1}{k!(rk + u)} \Gamma \left( \frac{1 + u + rk}{2} \right)\!\!\!\left( \frac{-p}{\omega^r} \right)^k\!\!\!.
\end{equation*}
Substituting \(x = y^2\)  and using~\eqref{temp5},
\begin{equation*}
I_2 = 2 \cdot \Phi\left(2a, b, 2n,  \sqrt{\frac{c}{2}}\right),\quad I_3 = 2 \cdot \Phi\left(2a, b, 2n, \sqrt{\frac{\delta}{2}}\right).
\end{equation*}
Finally, for $I_4$, we use the approximation~\cite{1210748},
\(
\text{erfc}(x)\approx \frac{1}{6} e^{-x^2} +\frac{1}{2} e^{-\frac{4x^2}{3}},
\)
to obtain
\begin{align}
I_{4} &\approx \sum_{j=1}^4D_j\int_0^\infty x^{a-1} \exp\left(-b x^{n} - C_j x\right) dx \approx \sum_{j=1}^4 \Psi(D_j,C_j), \label{temp6}
\end{align}
where $\mathbf{C}$, $\mathbf{D}$, and $\Psi(D_j,C_j)$  are given in the statement of Lemma~\ref{Theorem1}. Equation~\eqref{temp6} is due to~\cite[Eq.10]{11072420}~\cite{wolfram}. Combining $I_1$, $I_2$, $I_3$, and $I_4$ concludes the proof. 
\end{proof}
\section{Analysis of ZF \texorpdfstring{\( P_e \)}{Pe} in outdoor THz channels} 
\label{ZF_PC_MG}
Let \(z = \frac{n_{tr}}{h}  + \frac{n}{h} = \varphi + \omega\) be the ZF output of the channel given in~\eqref{Sp_Thz_sys}, where \(\varphi \sim \mathcal{CN}(0, 2\bar{\sigma}_{tr}^2)\) with \(2\bar{\sigma}_{tr}^2 \!=\! \kappa_t^2+\kappa_r^2\). 
The PDF of \(\omega\) in Cartesian coordinates is implied by~\eqref{temp3} 
\begin{align}
    & f_\omega(\omega_x, \omega_y) = \sum_{i=1}^{K} \frac{\alpha_i}{2\pi\sigma_n^2\nu^{\beta_i}}  \int_{0}^{\infty}y^{\beta_i+1}\exp\left(-\frac{(\omega_x^2+\omega_y^2) y^2}{2\sigma_n^2}-\frac{\zeta_i}{\nu} y\right) dy, \notag
\end{align}
which implies by convolution
\begin{align}
    f_{Z_y, Z_x}(z_y, z_x)
    & = \iint_{-\infty}^{\infty} f_\varphi(\varphi_y, \varphi_x)  
    f_\omega(z_x - \varphi_x, z_y - \varphi_y) \, d\varphi_y \, d\varphi_x\notag\\
    &= \iint_{-\infty}^{\infty} \frac{1}{2\pi \bar{\sigma}_{tr}^2} 
    \exp\left(-\frac{\varphi_x^2 + \varphi_y^2}{2\bar{\sigma}_{tr}^2}\right) 
    \frac{1}{2\pi\sigma_n^2} \sum_{i=1}^{K} \frac{\alpha_i}{\nu^{\beta_i}}  \int_{0}^{\infty}\! y^{\beta_i+1} 
    \exp\left(\!-\frac{(\|{\bf z} - \boldsymbol{ \varphi}\|^2 y^2}{2\sigma_n^2}  
    -\!\frac{\zeta_i}{\nu} y\!\right) dyd\varphi_yd\varphi_x\notag\\
    &= \sum_{i=1}^{K} \frac{\alpha_i}{4\pi^2 \bar{\sigma}_{tr}^2 \sigma_n^2 \nu^{\beta_i}} \int_{0}^{\infty} y^{\beta_i+1}  
    \exp\left(-\frac{\zeta_i}{\nu} y\right) 
    \iint_{-\infty}^{\infty} \exp\left(\!-\frac{\|\boldsymbol{\varphi}\|^2}{2\bar{\sigma}_{tr}^2} \!
    - \!\frac{\|{\bf z} - \boldsymbol{ \varphi}\|^2  y^2}{2\sigma_n^2}\!\right)d\varphi_y d\varphi_x dy,\label{eq:convolution_integral}
\end{align}
where $\|\boldsymbol{\varphi}\|^2 = \varphi_x^2 + \varphi_y^2$ and $\|{\bf z} - \boldsymbol{ \varphi}\|^2 = (z_x - \varphi_x)^2 + (z_y - \varphi_y)^2$. We first evaluate the inner integral
\begin{align}
    &I= \int\!\!\!\int_{-\infty}^{\infty}\!\exp\left(\!  - \frac{\|\boldsymbol{\varphi}\|^2}{2\bar{\sigma}_{tr}^2} -\frac{\|{\bf z} - \boldsymbol{\varphi}\|^2}{2\sigma_n^2} y^2\!\!\right)d\varphi_xd\varphi_y \notag \\
    &= \exp\left(-\frac{y^2}{2\sigma_n^2} (z_x^2 + z_y^2)\right) \int_{-\infty}^{\infty}\hspace{-2.5mm}\exp\!\left(\!\!-\!\left(\!\frac{y^2}{2\sigma_n^2}\!+\!\frac{1}{2\bar{\sigma}_{tr}^2}\!\right)\!\! \|\boldsymbol{\varphi}\|^2 \! + \! \frac{y^2}{\sigma_n^2}(z_x \varphi_x\!+\!z_y\varphi_y)\!\right)\!d\varphi_x d\varphi_y,\nonumber\\
    &= \frac{2\pi\bar{\sigma}_{tr}^2 \sigma_n^2}{(\bar{\sigma}_{tr}^2 + \frac{\sigma_n^2}{y^2})}\exp\left(-\frac{z_x^2 + z_y^2}{2(\bar{\sigma}_{tr}^2 + \frac{\sigma_n^2}{y^2})} \right),\nonumber
\end{align}
where we used~\cite[Eq 3.323.2]{zwillinger2014table} to write the last equality. Substituting the value of $I$ into~\eqref{eq:convolution_integral}, 
\begin{align*}
    f_{Z_x,Z_y}(z_x, z_y) &= \sum_{i=1}^{K} \frac{\alpha_i}{2\pi (\bar{\sigma}_{tr}^2 + \frac{\sigma_n^2}{y^2}) \nu^{\beta_i}} \int_{0}^{\infty} y^{\beta_i+1}\exp\left(-\frac{z_x^2 + z_y^2}{2(\bar{\sigma}_{tr}^2 + \frac{\sigma_n^2}{y^2})} - \frac{\zeta_i}{\nu} y\right) \, dy. 
\end{align*}
Finally, the average SER for the standard $4$-QAM symbol constellation can be computed by~\eqref{SER-ZF} as
\begin{align}
    &P_{e|s_1} = 1 - \sum_{i=1}^{K} \frac{\alpha_i}{2\pi (\bar{\sigma}_{tr}^2 + \frac{\sigma_n^2}{y^2}) \nu^{\beta_i}}  \iint_0^\infty \int_{0}^{\infty}y^{\beta_i+1}\exp\left(-\frac{(\tilde{r}_x - s_x)^2 + (\tilde{r}_y - s_y)^2}{2(\bar{\sigma}_{tr}^2 + \frac{\sigma_n^2}{y^2})} - \frac{\zeta_i}{\nu} y\right) \, dy \, d\tilde{r}_x \, d\tilde{r}_y \notag\\
    &= 1 -  \sum_{i=1}^{K} \frac{\alpha_i}{\nu^{\beta_i}} \int_{0}^{\infty} y^{\beta_i+1} \exp\left(-\frac{\zeta_i}{\nu} y\!\!\right) \left(\!1 - Q\left(\frac{s_y}{\sqrt{\bar{\sigma}_{tr}^2 + \frac{\sigma_n^2}{y^2}}}\!\right)\!\right)\!\left(\!1 - Q\left(\frac{s_x}{\sqrt{\bar{\sigma}_{tr}^2 + \frac{\sigma_n^2}{y^2}}}\!\right)\!\right)\!dy \notag.
\end{align}
For a 4-QAM constellation, \(P_{c|s_1}\) reduces to
\begin{align*}
P_{c|s_1}=\sum_{i=1}^{K}\!\frac{\alpha_i}{\nu^{\beta_i}} \int_{0}^{\infty}\hspace{-1.5mm}y^{\beta_i+1}\!\left(1\!-\!Q\!\left(\!\frac{s_x}{\sqrt{\bar{\sigma}_{tr}^2\!\!+\!\!\frac{\sigma_n^2}{y^2}}}\!\right)\!\right)^2\!\!\exp\left(\!\!-\frac{\zeta_i}{\nu}\!y\right)\!dy.
\end{align*}
\section{SER Calculation in Circular Additive Noise and Fading Channels}\label{appB}
In this appendix, we show that computing the SER of the fading channel by averaging the AWGN SER over the fading distribution is valid only when the channel and the additive noise are independent. Let $r = hs + n$, where $h$ and $n$ are respectively the channel multiplicative fading and the additive noise which are assumed to be circular and having a joint density $f_{h,n}(a,b)$. Define the ratio \( z  = \frac{n}{h}  = \frac{\rho_n}{\rho_h} e^{i(\theta_n - \theta_h)} = \rho e^{i\theta} \). We first note that because both $\theta_n$ and $\theta_h$ are uniform $\mathcal{U}(0,2\pi)$, $\theta \sim \mathcal{U}(0,2\pi)$. Hence, $\rho$ and $\theta$ are independent. The probability of correct detection, \( P_c \), given a symbol $s$ is transmitted is
\begin{align}
    P_c &= 1 - P_e = \iint_{D_{s}} f_Z(x, y) \, dx \, dy \notag\\
 &= \iint_{D_{s}} f_{\rho, \theta}(\rho, \theta) \cdot \rho \, d\rho \, d\theta \notag\\
    &= \frac{1}{2\pi} \iint_{D_{s}} f_{\rho}(\rho) \cdot \rho \, d\rho \, d\theta \label{eq:ind}\\
    &= \frac{1}{2\pi} \iint_{D_{s}} \rho \int_{-\infty}^{\infty} |a| f_{\rho_n,\rho_h}(\rho a, a) \, da \, d\rho \, d\theta, \label{eq:meth1}
\end{align}
where ${D_{s}}$ is the ML decision region corresponding to symbol $s$ and where we used the fact that \( \rho \) and \( \theta \) are independent in~(\ref{eq:ind}). Equation~(\ref{eq:meth1}) is justified because 
\begin{equation*}
    f_{\rho}(\rho) = \int_{-\infty}^{\infty} |a| f_{\rho_n, \rho_h}(\rho a, a) \, da.
\end{equation*}
When considering~(\ref{eq:meth1}) under the special case of \( h \) and \( n \) being independent, we can further simplify the expression to
\begin{align}
    P_c &= \frac{1}{2\pi} \iint_{D_{s}} \rho \int_{-\infty}^{\infty} |a| f_{\rho_h}(a) f_{\rho_n}(\rho a) \, da \, d\rho \, d\theta \notag\\
    &= \frac{1}{2\pi} \int_{-\infty}^{\infty} |a| f_{\rho_h}(a) \iint_{D_{s}}  f_{\rho_n}(\rho a) \rho \, d\rho \, d\theta \, da \notag\\
    &= \int_{-\infty}^{\infty} f_{\rho_h}(a) \frac{1}{2\pi} \iint_{D_{s}} |a|  f_{\rho_n}(\rho a) \rho \, d\rho \, d\theta \, da \notag\\
    &=  \int_{-\infty}^{\infty} f_{\rho_h}(a) \frac{1}{2\pi} \iint_{D_{s}} f_{\frac{\rho_n}{a}}(\rho) \rho\, d\rho \, d\theta \, da \notag\\
    &= \int_{-\infty}^{\infty} f_{\rho_h}(a) \iint_{D_{s}} f_{Z||h|=a}(x,y) \, dx \, dy \, da \notag\\
    &= \int_{-\infty}^{\infty} f_{\rho_h}(a) P_c(a) \, da, \label{eq:meth2}
\end{align}
where $P_c(a) =  1 - P_e(a) = \iint_{D_{s}} f_{(Z |\abs{h}=a)}(x,y) \, dx \, dy$ is the probability of correct detection given $|h| = a$.
This shows that the use of equation~\eqref{eq:meth2} to compute the SER is only valid whenever $h$ and $n$ are independent. 
\newpage

\bibliographystyle{IEEEtran}
\bibliography{IEEEabrv,reference.bib}
\end{document}